\documentclass[twocolumn,amsthm]{autart}

\usepackage{amsmath}
\usepackage{amssymb}
\usepackage{enumitem}
\usepackage{graphicx}
\usepackage{mathtools}
\usepackage{nicefrac}
\usepackage{subcaption}
\usepackage{url}

\begin{document}

\begin{frontmatter}
    \runtitle{Risk-Aware Control of Nonlinear Systems}
    \title{Risk-Aware Control of Systems with\\Quasi-Cone-Bounded Nonlinearities
    \thanksref{titlefootnote}}
    \thanks[titlefootnote]{This article is based on the MASc thesis by D. Patel \cite{patel2025risk}. It was not presented at an IFAC meeting and was not published in any conference proceedings. This research was supported by the Edward S. Rogers Sr. Department of Electrical and Computer Engineering at the University of Toronto, and by the Natural Sciences and Engineering Research Council of Canada (NSERC) Discovery Grants Program [RGPIN-2022-04140]. Cette recherche a été financée par le Conseil de recherches en sciences naturelles et en génie du Canada (CRSNG).\\ $^\ast$\hspace{1.5mm}Corresponding author: D. Patel (email address: dhairya.patel@mail.utoronto.ca).}

    \author[UofT]{Dhairya Patel$^{\ast\text{,}}$}
    and
    \author[UofT]{Margaret P. Chapman}
    \address[UofT]{The Edward S. Rogers Sr. Department of Electrical and Computer Engineering, University of Toronto, 10 King's College Road, Toronto, Ontario, Canada M5S 3G4 }

    \begin{abstract}
        We develop a tractable, rigorous approach to risk-aware control for a class of nonlinear systems. While many classical control methods reduce uncertainty to a simple average or a worst-case outcome, risk-aware control aims to equip systems with a refined awareness of uncertainty. Efficient methods for risk-aware control of linear systems are available, but there is a paucity of tools for tractable, risk-aware control of nonlinear systems. To bridge this gap, we develop an analytical, suboptimal controller with respect to a risk-aware performance criterion for systems with nonlinearities characterized by cone-like bounds. Numerical examples demonstrate benefits of the characterization of nonlinearities and risk that we consider.
    \end{abstract}

\end{frontmatter}

\section{Introduction}

The behavior of systems encountered in practice is inherently uncertain. One approach to modeling such systems is by representing the uncertainty stochastically via random variables. Examples of stochastic systems appear in areas such as finance \cite{stein2012stochastic}, energy \cite{habib2023stochastic}, and robotics \cite{vahs2024risk}. It is standard to summarize system performance in terms of a sum of stage costs over time, which we call a standard cost. As a result of the stochasticity in the system model, a standard cost is a random variable which cannot be directly optimized \cite[p. 72]{kumar1986stochastic}. As such, the classical \emph{risk-neutral} approach to stochastic optimal control optimizes the expectation of the standard cost \cite{kumar1986stochastic,hernandezlerma2012discrete}. Although optimal in the ``average'' case, this approach may fail to consider rare, yet significant outcomes that may be dangerous to safety-critical systems \cite{majumdar2019how,smith2023on}. An alternative is the \emph{robust} approach for nonstochastic uncertain systems, which optimizes the worst-case setting, but may be too pessimistic \cite{zhou1996robust,majumdar2019how,smith2023on}. These limitations demonstrate the need for a decision-making framework that enables the user to consider both ordinary and more detrimental settings systematically \cite{majumdar2019how,smith2023on,chapman2025riskaware,akella2025risk}. This motivates the \emph{risk-aware} framework, which uses risk measures to extract crucial characteristics from the cost distribution to analyze or optimize \cite{majumdar2019how,shapiro2021lectures,wang2022risk,smith2023on,chapman2025riskaware,royset2025Risk,akella2025risk}. Informally, a \emph{risk measure} is a map from a space of random variables into the extended real line \cite[Sect. 6.3]{shapiro2021lectures}. There are many examples, such as variance, expectation-variance, distributionally robust expectations, exponential utility, and conditional value-at-risk \cite{shapiro2021lectures}. We discuss some of these examples and other ways to measure risk in the context of control theory below.\footnote{While this brief paper does not concern distributionally robust control, we note that optimizing a distributionally robust expectation can be viewed as a risk-aware optimization problem because a distributionally robust expectation is a type of risk measure.}

An early approach to risk-aware control is with exponential utility, in which the expectation of the exponential of a scalar multiple of the standard cost is considered \cite{whittle1981risk}. The exponential utility amplifies the penalty for potential increases in the cost to promote more risk-averse controllers. A method to risk-aware control that measures risk similarly to exponential utility was developed for linear systems with multiplicative noise \cite{ito2018risk}; this method was later generalized to other controller types for linear systems \cite{ito2025weighted}. Also related to the exponential utility approach is mixed $\mathcal{H}_2/\mathcal{H}_\infty$ control \cite{bernstein1989LQG,geromel1995convex,zhang2021policy}. Mixed $\mathcal{H}_2/\mathcal{H}_\infty$ control relates to the problem of optimizing the expected standard cost for a linear system with additive Gaussian noise subject to a robust stability constraint, an example of integrating the average-case and worst-case frameworks \cite{bernstein1989LQG,geromel1995convex,zhang2021policy}.

A more recent method involves optimizing the expected standard cost while constraining the predictive variance of the state cost $x_t^\top Q x_t$  \cite{tsiamis2020riskconstrained}. This method yields an optimal affine state-feedback controller that considers the covariance and skew of the additive noise \cite{tsiamis2020riskconstrained}. A closed-form solution can be determined for a large class of disturbances \cite{tsiamis2020riskconstrained}, rather than just Gaussian disturbances as in the classical exponential utility approach \cite{whittle1981risk}. The predictive variance approach to risk-aware control has attracted attention, as elaborated in \cite{weitobesubmitted}, which shares authorship with the current article. The work \cite{weitobesubmitted} extends predictive variance to a temporally coupled cost, and derives an optimal controller that penalizes temporal and stochastic variabilities of the state trajectory for a linear system with additive noise. In contrast, the current article develops an approach to risk-aware control using a different type of risk assessment that concerns multiplicative noise for a class of nonlinear systems, and does not consider temporal coupling (see Contribution paragraph for details).

Another popular approach involves the use of conditional value-at-risk (CVaR) from finance \cite{shapiro2021lectures,vanParys2016distributionally,kishida2022risk,chapman2023on}. CVaR is conceptually appealing because it represents the expectation of a random variable in the worst $q \cdot 100\%$ of outcomes for a given choice of $q \in (0, 1)$ \cite{shapiro2021lectures}, \cite[Sect. 3.1.5]{wang2022risk}. As $q$ stands for a fraction of the worst cases, CVaR presents an approach to tune quantitatively the level of risk between the risk-neutral and worst-case settings \cite{shapiro2012minimax,majumdar2019how}. Also, CVaR possesses properties that enable it to have an interpretation for and application to distributionally robust control \cite{vanParys2016distributionally,shapiro2021lectures,kishida2022risk}.

However, many of these approaches are designed for linear or affine systems, thereby limiting their applicability to nonlinear systems \cite{whittle1981risk,bernstein1989LQG,geromel1995convex,vanParys2016distributionally,ito2018risk,tsiamis2020riskconstrained,zhang2021policy,kishida2022risk,ito2025weighted,weitobesubmitted}. While some risk-aware controllers for linear or affine systems can be found analytically via algebraic Riccati-like equations, e.g., \cite{whittle1981risk,fujimoto2011optimalfinitetime,ito2018risk,tsiamis2020riskconstrained,kishida2022risk,weitobesubmitted}, methods to compute risk-aware controllers for nonlinear systems often involve dynamic programming without an analytical form. Examples include algorithms to optimize Markov decision processes, with performance quantified via dynamic risk measures \cite{ruszczynski2010riskaverse}, spectral risk measures \cite{bauerle2021minimizing}, and CVaR \cite{bauerle2011markov,chapman2023on}. Dynamic programming can be intractable for large state spaces \cite{bertsekas1996neuro}, especially when extended state spaces are required \cite{bauerle2011markov,bauerle2021minimizing,chapman2023on,smith2023on}. This presents a gap in the tractable risk-aware control of nonlinear systems.

To bridge this gap, we look to the control of systems with \emph{cone-bounded} nonlinearities, which the authors of \cite{scherzinger1982estimation} consider in a risk-neutral, partially observable setting. A system in this class behaves approximately like a nominal linear system, with its deviation from the nominal system being contained in a cone of a certain form \cite{scherzinger1982estimation}. A cone-bounded function on $\mathbb{R}^n$ is also sector-bounded and vice versa (to be discussed in Section \ref{sec:cone_bounded_dynamics}). Examples on $\mathbb{R}$ include saturation functions, such as $f(x) \triangleq x$ if $|x| < x_\text{sat}$, $f(x) \triangleq \operatorname{sgn}(x) x_\text{sat}$ otherwise, with $x_\text{sat} > 0$, and quantizers (to be presented in Sections \ref{sec:quasi_cone_bounded_dynamics} and \ref{sec:numerical_results}).

In this article, we study a \emph{risk-aware} control problem for a class of fully observable systems with more general nonlinearities. We consider nonlinearities that satisfy cone-like bounds using a simple relaxation of cone bounds in a positive-semi-definite form; we call these nonlinearities \emph{quasi-cone-bounded}. The systems have multiplicative noise (e.g., a random matrix multiplies the state) and additive noise, similar to those in \cite{scherzinger1982estimation}. However, only multiplicative noise turns out to play a central role in the risk assessment we consider (to be made precise in Proposition \ref{prop:quadform}). We use a risk assessment similar to the one in \cite{fujimoto2011optimalfinitetime}, which the authors developed for linear systems, since their model includes both types of noise. The authors of \cite{fujimoto2011optimalfinitetime} assess risk in terms of \emph{state variability} via the predictive covariance of $x_{t+1}$, and derive an optimal linear state-feedback controller whose gains depend on multiplicative noise statistics. The authors of \cite{jia2024decentralized} use a risk assessment resembling the one in \cite{fujimoto2011optimalfinitetime}, but without conditioning, in a decentralized, partially observable, linear setting. In contrast, many other approaches to risk-aware control of linear or affine systems use models with multiplicative noise or additive noise but not both \cite{whittle1981risk,vanParys2016distributionally,ito2018risk,tsiamis2020riskconstrained,kishida2022risk,ito2025weighted,weitobesubmitted}.

\textbf{Contribution.} We develop a tractable, theoretically rigorous approach to \emph{risk-aware} control for a class of \emph{nonlinear} systems. Specifically, we consider discrete-time, fully observable systems with \emph{quasi-cone-bounded} nonlinearities subject to multiplicative and additive noise on a finite horizon. Using a performance criterion similar to that in \cite{fujimoto2011optimalfinitetime}, we propose a suboptimal controller $\hat{\mathbf{u}}$ that enjoys \emph{all} of the following features:
\begin{enumerate}
    \item[(a)] \emph{Nonlinear dynamics:} $\hat{\mathbf{u}}$ is designed for a class of nonlinear stochastic systems.
    \item[(b)] \emph{Analytical form:} $\hat{\mathbf{u}}$ has a Riccati-like form similar to that of LQR controllers.
    \item[(c)] \emph{Risk-awareness:} $\hat{\mathbf{u}}$ is designed to reduce variability of the state arising from multiplicative noise.
\end{enumerate}
To develop this contribution, first we show how cone-boundedness is equivalent to other characterizations (Propositions \ref{prop:raymondcone} and \ref{prop:sectorcone}), and then present the extension to quasi-cone-boundedness. Next, we prove that the risk-aware term in the performance criterion enjoys a quadratic representation in the state and control under a general condition on the nonlinearity (Proposition \ref{prop:quadform}). This representation and an upper bound of a quadratic form of a quasi-cone-bounded function (Lemma \ref{lem:quadratic_bound}) underpin $\hat{\mathbf{u}}$, which we present and analyze in Theorem \ref{thm:regulation_upper_bounds}. We also show that a special case of $\hat{\mathbf{u}}$ exactly solves a risk-aware optimal linear control problem (Lemma \ref{lemma12}). We provide an example system with a cone-bounded nonlinearity in which $\hat{\mathbf{u}}$ numerically exhibits effective regulation over a range of nonlinearities. Lastly, we numerically evaluate the viability of our approach in a high-dimensional, quasi-cone-bounded setting.

\textbf{Organization.} Section \ref{sec:problem_formulation} presents the preliminary problem formulation. Cone-bounded and quasi-cone-bounded functions are studied in Sections \ref{sec:cone_bounded_dynamics} and \ref{sec:quasi_cone_bounded_dynamics}, respectively. The risk-aware optimal control problem introduced in Section \ref{sec:problem_formulation} is specialized to systems with quasi-cone-bounded nonlinearities in Section \ref{sec:quasi_cone_bounded_dynamics}. Section \ref{sec:reformulating_variance_suppression} reformulates the risk-aware objective. Section \ref{sec:suboptimal_controller} presents and studies a risk-aware controller $\hat{\mathbf{u}}$ that is suboptimal with respect to the specialized problem. Sections \ref{sec:numerical_results} and \ref{sec:conclusion} provide numerical results and concluding remarks, respectively. We include supporting, existing knowledge in the Appendix (Lemmas \ref{lem:peter_paul_inequality}--\ref{lemma:opt}).

\textbf{Notation.} $\mathbb{N} \triangleq \{1,2,3,\dots\}$, and $\mathbb{R}$ is the real line. $\mathbb{R}^{n \times m}$ is the set of real $n \times m$ matrices. For column vectors $v_1,\dots,v_m$, we define $(v_1, \ldots, v_m) \triangleq [v_1^\top \  \cdots \ v_m^\top]^\top$. The column vector form of $M  \in \mathbb{R}^{n \times m}$ is $\text{vec}(M) \triangleq (M_1,\dots,M_m)$, where $M_j$ is the $j$th column of $M$. $I_n$ is the $n \times n$ identity matrix. $0_{n \times m}$ is the $n \times m$ zero matrix, with $0_n \triangleq 0_{n \times 1}$. For $M \in \mathbb{R}^{n \times n}$ and $v \in \mathbb{R}^n$, we define $(*)^\top M v \triangleq v^\top M v$. $\|\cdot\|$ is the Euclidean norm on $\mathbb{R}^p$ for some $p \in \mathbb{N}$, or the induced norm on $\mathbb{R}^{n \times m}$ (induced by the Euclidean norms on $\mathbb{R}^m$ and $\mathbb{R}^n$) \cite[Def. 11.4.1, p. 841]{bernstein2018scalar}. $\mathbb{S}^n$ is the set of real symmetric $n \times n$ matrices. $\mathbb{S}^n_{\geq}$ ($\mathbb{S}^n_{>}$) is the set of real symmetric positive semidefinite (definite) $n \times n$ matrices. tr$(\cdot)$ denotes trace. $\operatorname{sgn}(\cdot)$ ($\lfloor \cdot \rfloor$, $\lceil \cdot \rceil$) is the sign (floor, ceiling) function. $(\Omega, \mathcal{F}, \mathbb{P})$ is a probability space upon which all random elements in this article are defined. $\mathbb{E}(\cdot)$ ($\mathbb{E}(\cdot|\cdot)$) and cov$(\cdot)$ (cov$(\cdot|\cdot)$) denote (conditional) expectation and covariance, respectively. a.e. means almost everywhere with respect to $\mathbb{P}$. A random matrix $M$ is called square-integrable if $\mathbb{E}(\|M\|^2)$ is finite. Given measurable spaces $(\Omega, \mathcal{F})$ and $(\Omega', \mathcal{F}')$, $\phi : (\Omega, \mathcal{F}) \rightarrow (\Omega', \mathcal{F}')$ means that $\phi : \Omega \rightarrow \Omega'$ is measurable relative to $\mathcal{F}$ and $\mathcal{F}'$ \cite[p. 35]{ash1972probability}. $\mathcal{B}_{\mathbb{R}^n}$ is the Borel $\sigma$-algebra on $\mathbb{R}^n$.

\section{Preliminary Model and Problem}
\label{sec:problem_formulation}

Consider a stochastic system subject to multiplicative and additive disturbances on a horizon $\mathbb{T} \triangleq \{0, \ldots, T - 1\}$ of length $T \in \mathbb{N}$, with an $\mathbb{R}^n$-valued, random initial state $x_0$. Let the random vectors $u_t$ and $v_t$ be $\mathbb{R}^m$-valued control and $\mathbb{R}^n$-valued additive disturbance inputs at time $t \in \mathbb{T} $, respectively. Let the random matrices $\tilde{A}_t$ and $\tilde{B}_t$ be $\mathbb{R}^{n \times n}$-valued and $\mathbb{R}^{n \times m}$-valued multiplicative disturbances at time $t \in \mathbb{T}$, respectively. For each $t \in \mathbb{T}$, the state at time $t+1$ is defined by
\begin{equation}
    \label{eq:dynamics}
    x_{t + 1} = f_t(x_t, u_t) + \tilde{A}_t x_t + \tilde{B}_t u_t + v_t,
\end{equation}
where $f_t : \mathbb{R}^n \times \mathbb{R}^m \to \mathbb{R}^n$ is a Borel-measurable nonlinear function. Denote the history at time $t$ as $h_t \triangleq (x_0, u_0, x_1, \ldots, u_{t - 1}, x_t)$ for each $t \in \mathbb{T}$, with $h_0 \triangleq x_0$. Then, define the $\sigma$-algebra generated by the history as $\mathcal{F}_t \triangleq \sigma(h_t)$. For convenience, we also express the disturbances as $w_t \triangleq (\text{vec}(\tilde{A}_t), \text{vec}(\tilde{B}_t), v_t)$. Hence, the system \eqref{eq:dynamics} can be written as $x_{t+1} = \bar{f}_t(x_t,u_t,w_t)$ for some Borel-measurable function $\bar{f}_t$. We assume that the system has the form of \eqref{eq:dynamics}, although we may not explicitly state this assumption. We do not require full distributional knowledge but use the following:
\begin{assum}[Initial State, Disturbances]
    \label{assum:initial_state}
    (a) $x_0$ is square-integrable;
    (b) $x_0, w_0, \dots, w_{T-1}$ are independent;
    (c) $v_t$ is independent of $[\tilde{A}_t \; \tilde{B}_t]$;
    (d) $v_t$, $\tilde{A}_t$, and $\tilde{B}_t$ are square-integrable and zero-mean; (e) $\mathbb{E}(x_0)$, $\mathbb{E}(x_0 x_0^\top)$, $\mathbb{E}(w_t w_t^\top)$ are known. Items (c)--(e) apply to each $t \in \mathbb{T}$.
\end{assum}

\begin{assum}[Nonlinearities]
    \label{assum:dynamics}
    For each $t \in \mathbb{T}$, $f_t$ in \eqref{eq:dynamics} satisfies the following: if $x$ and $u$ are square-integrable random vectors, then $f_t(x, u)$ is square-integrable.
\end{assum}

\begin{assum}[Controls]
    \label{assum:control}
    For each $t \in \mathbb{T}$, $u_t$ is square-integrable and $\mathcal{F}_t$-measurable, i.e., $u_t \in \mathcal{L}^2(\mathcal{F}_t)$.
\end{assum}

$u_t$ being $\mathcal{F}_t$-measurable for each $t \in \mathbb{T}$ ensures that the controller $\mathbf{u} \triangleq (u_0,\dots,u_{T-1})$ depends causally on the system history. This condition and $x_0,w_0,\dots,w_{T-1}$ being independent are standard conditions for systems of the form $x_{t+1} = \bar{f}_t(x_t,u_t,w_t)$, guaranteeing that $w_t$ is independent of $h_t$ for each $t \in \mathbb{T}$ and that the system is Markovian \cite[Lemma 6.6, p. 18]{kumar1986stochastic}. Assumptions~\ref{assum:initial_state}--\ref{assum:control} ensure that $x_0,\dots,x_T$ are square-integrable and that the objective to be defined in \eqref{eq:objective1} is real-valued (use Proposition \ref{prop:quadform}). The disturbances being zero-mean simplifies analysis and does not restrict us. We can reformulate the dynamics in terms of zero-mean disturbances, such that the reformulated nonlinearity does not violate Assumption \ref{assum:dynamics} (or the stronger quasi-cone-boundedness assumption to be introduced in Section \ref{sec:quasi_cone_bounded_dynamics}).

Under Assumptions \ref{assum:initial_state}--\ref{assum:control}, we define the following risk-aware objective functional to be considered in this work:
\begin{equation}
    \label{eq:objective1}
    \mathcal{J}_{\mathbf{u}} \triangleq \mathbb{E}\left( \textstyle c_T(x_T) + \sum_{t=0}^{T-1} c_t(x_t,u_t)  + \varrho_t \right),
\end{equation}
where $\varrho_t$ and $c_t(x_t,u_t)$ are random variables corresponding to the risk-aware and risk-neutral parts of $\mathcal{J}_{\mathbf{u}}$, respectively. For each $t \in \mathbb{T}$, $\mathbb{E}(\varrho_t)$ is a variance-like cost similar to \cite[eq. (6)]{fujimoto2011optimalfinitetime} with
\begin{equation}
    \label{eq:riskrandomvariable}
    \varrho_t \triangleq \operatorname{tr}\left(Z_t \operatorname{cov}(x_{t+1}  |  \mathcal{F}_t)\right)
\end{equation}
for a design parameter $Z_t \in \mathbb{S}^n_\geq$; intuitively, $\mathbb{E}(\varrho_t)$ assesses risk in terms of the variability of $x_{t+1}$ estimated from the system history up to time $t$. For each $t \in \mathbb{T}$, $c_t$ is
\begin{equation}
    \label{eq:costs}
    c_t(x, u) \triangleq \genfrac{[}{]}{0pt}{0}{x}{u}^\top \genfrac{[}{]}{0pt}{0}{Q_t \;\: S_t}{S_t^\top \;\: R_t} \genfrac{[}{]}{0pt}{0}{x}{u}, \quad  (x,u) \in \mathbb{R}^n \times \mathbb{R}^m,
\end{equation}
where $Q_t \in \mathbb{S}^n_\geq$, $R_t \in \mathbb{S}^m_>$, and $S_t \in \mathbb{R}^{n \times m}$ are chosen so that the matrix in \eqref{eq:costs} is in $\mathbb{S}^{n + m}_\geq$; $c_T(x) \triangleq x^\top Q_T x$ for every $x \in \mathbb{R}^n$, where $Q_T \in \mathbb{S}^n_\geq$. The preliminary optimal control problem is then
\begin{equation}
    \label{eq:mean_field_optimal_control_problem}
    \begin{alignedat}{2}
        &\inf_{\mathbf{u}}       && \mathcal{J}_{\mathbf{u}} \\
        &\operatorname{s.t. } \; && x_{t + 1} = f_t(x_t, u_t) + \tilde{A}_t x_t + \tilde{B}_t u_t + v_t \text{ and}\\
        &                        && \text{$u_t \in \mathcal{L}^2(\mathcal{F}_t)$ for each $t \in \mathbb{T}$; Assumptions \ref{assum:initial_state}--\ref{assum:dynamics}}.
    \end{alignedat}
\end{equation}
This is a well-posed, interpretable problem but difficult to work with in its present form. Assumption \ref{assum:dynamics} permits a general nonlinearity $f_t$ with no immediately discernible structure to simplify the problem. Also, $\varrho_t$ \eqref{eq:riskrandomvariable} is defined in terms of $\operatorname{cov}(x_{t+1}|\mathcal{F}_t)$ instead of $(x_t,u_t)$ alone. While this definition gives $\varrho_t$ a risk-aware interpretation, it differs from that of a typical random stage cost. To circumvent the first issue, we introduce a stricter, structural assumption on $f_t$, permitting it to be \emph{quasi-cone-bounded} (Section \ref{sec:quasi_cone_bounded_dynamics}). To circumvent the second issue, we reformulate $\varrho_t$ (Section \ref{sec:reformulating_variance_suppression}). This article addresses a special case of \eqref{eq:mean_field_optimal_control_problem} when $f_t$ is quasi-cone-bounded and provides an analytical, suboptimal solution to the specialized problem (Section \ref{sec:suboptimal_controller}). Proceeding stepwise, next we present the simpler notion of cone-bounded functions and then the extension to quasi-cone-bounded functions.

\section{Cone-Bounded Functions}
\label{sec:cone_bounded_dynamics}

A function $f_t : \mathbb{R}^n \times \mathbb{R}^m \rightarrow \mathbb{R}^n$ is said to be \emph{cone-bounded} if there exist $A_t \in \mathbb{R}^{n \times n}$, $B_t \in \mathbb{R}^{n \times m}$, and $\Delta_t \in \mathbb{S}^{n + m}_\geq$ such that
\begin{equation}
    \label{eq:cone_bounded_condition}
    \left\|
        f_t(x, u) - \bigl[A_t \;\: B_t\bigr] \genfrac{[}{]}{0pt}{0}{x}{u}
    \right\|^2 \leq  \genfrac{[}{]}{0pt}{0}{x}{u}^\top \Delta_t  \genfrac{[}{]}{0pt}{0}{x}{u}
\end{equation}
for every $(x,u) \in \mathbb{R}^n \times \mathbb{R}^m$. $A_t$, $B_t$, and $\Delta_t$ in \eqref{eq:cone_bounded_condition} are called \emph{cone parameters}. This definition differs mildly from the condition presented in \cite[p. 34, eq. (3a)]{scherzinger1982estimation}, which we restate for easier comparison to \eqref{eq:cone_bounded_condition}: there exist $A_t \in \mathbb{R}^{n \times n}$, $B_t \in \mathbb{R}^{n \times m}$, $a_t \geq 0$, and $b_t \geq 0$ such that
\begin{equation}
    \label{eq:cone_bounded_condition_old}
    \left\|
        f_t(x, u) - \bigl[A_t \;\: B_t\bigr]  \genfrac{[}{]}{0pt}{0}{x}{u}
    \right\| \leq
    a_t \|x\| + b_t \|u\|
\end{equation}
for every $(x,u) \in \mathbb{R}^n \times \mathbb{R}^m$. The conditions in \eqref{eq:cone_bounded_condition} and \eqref{eq:cone_bounded_condition_old} have the same interpretation. That is, $A_t$ and $B_t$ characterize a nominal linear model that approximates $f_t$, i.e., $f_t(x,u) \approx A_t x + B_t u$, whereas $\Delta_t$ (respectively, $(a_t, b_t)$) characterizes how $f_t$ may deviate from this nominal model \cite[p. 34]{scherzinger1982estimation}. Unsurprisingly, these conditions are equivalent, as shown next. We justify our characterization (\ref{eq:cone_bounded_condition}) after the proposition.
\begin{prop}
    \label{prop:raymondcone}
    Let $f_t : \mathbb{R}^n \times \mathbb{R}^m \rightarrow \mathbb{R}^n$ be given. Then, there exist $A_t \in \mathbb{R}^{n \times n}$, $B_t \in \mathbb{R}^{n \times m}$, and $\Delta_t \in \mathbb{S}^{n + m}_\geq$
    such that (\ref{eq:cone_bounded_condition}) holds for every $(x,u) \in \mathbb{R}^n \times \mathbb{R}^m$ if and only if there exist $A_t \in \mathbb{R}^{n \times n}$, $B_t \in \mathbb{R}^{n \times m}$, $a_t \geq 0$, and $b_t \geq 0$ such that (\ref{eq:cone_bounded_condition_old}) holds  for every
    $(x,u) \in \mathbb{R}^n \times \mathbb{R}^m$.
\end{prop}
\begin{proof}
    \allowdisplaybreaks
    Define $z \triangleq f_t(x, u) - A_t x - B_t u$ and $\xi \triangleq (x,u)$. \emph{Part 1 $(\implies)$:} Using the Cauchy-Schwarz inequality and $\| \Delta_t \xi \| \leq \| \Delta_t\| \|\xi \|$, we have $\| z\|^2 \leq \| \Delta_t \| \|\xi\|^2$. Then, use $\|\xi\|^2 = \|x\|^2 + \|u\|^2 \leq (\|x\| + \|u\|)^2$ and take the square root of $\|z\|^2 \leq \|\Delta_t\| (\|x\| + \|u\|)^2$ to conclude that \eqref{eq:cone_bounded_condition_old} holds with $a_t = b_t = \|\Delta_t\|^{\nicefrac{1}{2}}$. \emph{Part 2 $(\impliedby)$:} Squaring both sides of (\ref{eq:cone_bounded_condition_old}) and applying Lemma \ref{lem:peter_paul_inequality} leads to $\|z\|^2 \leq   2 (a_t \|x\|)^2 + 2 (b_t \|u\|)^2 = \xi^\top \Delta_t \xi$, with $\Delta_t = 2\text{diag}( a_t^2 I_n, b_t^2 I_m)$.
\end{proof}

Despite the conditions in \eqref{eq:cone_bounded_condition}  and \eqref{eq:cone_bounded_condition_old} being equivalent, our condition in \eqref{eq:cone_bounded_condition} can permit increased specificity in the bounding region when $n > 1$ or $m > 1$. Let us consider a simple example. Suppose that $f$ is linear in all but the $i$th state, i.e., $f(x,u) = A x + B u + g(x_i)$, and \eqref{eq:cone_bounded_condition} holds with $\| f(x,u) -A x - B u \|^2 = \|g(x_i)\|^2 \leq \alpha x_i^2$. If $x_i =0$, then $\alpha x_i^2$ evaluates to zero, accurately indicating zero deviation from the nominal linear model. Conversely, the condition in \eqref{eq:cone_bounded_condition_old} would hold with $\| f(x,u) -A x - B u \| = \|g(x_i)\| \leq \sqrt{\alpha} \|x \|$; even if $x_i =0$, $\sqrt{\alpha} \|x \|$ would be nonzero in general, suggesting possible nonlinear behavior when there is none. Our characterization \eqref{eq:cone_bounded_condition} enables us to propose a suboptimal controller that can achieve more effective regulation, as illustrated numerically in Section \ref{sec:numerical_results}.

Our cone-boundedness condition on $\mathbb{R}^n$ is also equivalent to a sector-boundedness condition that closely resembles the one in \cite[Def. 6.2, eq. (6.5)]{khalil2002nonlinear}.
\begin{prop}
    \label{prop:sectorcone}
    Let $f : \mathbb{R}^n \rightarrow \mathbb{R}^n$ be given. Then, there exist $A \in \mathbb{R}^{n \times n}$ and $\Delta \in \mathbb{S}_{\geq}^n$ s.t. $\|f(x) - A x\|^2 \leq x^\top \Delta x$ for every $x \in \mathbb{R}^n$ if and only if there exists $F \triangleq F_{2} - F_{1} \in  \mathbb{S}_{\geq}^n$ s.t. $(f(x) - F_{1} x)^\top (f(x) - F_{2} x) \leq 0$ for every $x \in \mathbb{R}^n$.
\end{prop}
\begin{proof}
    Based on techniques from \cite[p. 232]{khalil2002nonlinear}. \emph{Part 1 $(\implies)$:} Choose $F_1 \triangleq A - \|\Delta\|^{\nicefrac{1}{2}} I_n$ and $F_2 \triangleq A +\|\Delta\|^{\nicefrac{1}{2}} I_n$. \emph{Part 2 $(\impliedby)$:} Choose $A \triangleq(F_1 + F_2)/2$ and $\Delta \triangleq (F_2 - F_1)^\top (F_2 - F_1) / 4$.
\end{proof}

\section{Quasi-Cone-Bounded Functions}
\label{sec:quasi_cone_bounded_dynamics}

Having studied how the cone-bounded condition in \eqref{eq:cone_bounded_condition} is related to some existing conditions, we consider a simple relaxation. A function $f_t : \mathbb{R}^n \times \mathbb{R}^m \rightarrow \mathbb{R}^n$ is said to be \emph{quasi-cone-bounded} if there exist $A_t \in \mathbb{R}^{n \times n}$, $B_t \in \mathbb{R}^{n \times m}$, $\Delta_t \in \mathbb{S}^{n + m}_\geq$, and $\delta_t \geq 0$ such that
\begin{equation}
    \label{eq:quasi_cone_bounded_condition}
    \left\|
        f_t(x, u) - f_t^0 - \bigl[A_t \;\ B_t\bigr] \genfrac{[}{]}{0pt}{0}{x}{u}
    \right\|^2 \!\leq\!  \genfrac{[}{]}{0pt}{0}{x}{u}^\top \Delta_t \genfrac{[}{]}{0pt}{0}{x}{u} + \delta_t
\end{equation}
for every $(x,u) \in \mathbb{R}^n \times \mathbb{R}^m$, with $f_t^0 \triangleq f_t(0_n, 0_m)$. We refer to $f_t^0$, $A_t$, $B_t$, $\Delta_t$, and $\delta_t$ as \emph{quasi-cone parameters}. This condition differs from (\ref{eq:cone_bounded_condition}) in just two ways: the inclusion of $f_t^0$ and $\delta_t$.
Let us illustrate how introducing such terms generalizes cone-boundedness using the functions $f_1, f_2, f_3 : \mathbb{R} \rightarrow \mathbb{R}$ shown in Fig. \ref{fig:example_bounding_regions} (see caption for their definitions). Functions similar to $f_1$ appear in quantized feedback control, e.g., see \cite[Figure 1]{fu2005sector}.
Using $z-1 < \lfloor z \rfloor \leq z$ for any $z \in \mathbb{R}$ and $f_1$ being an odd function, one can show that $f_1$ is cone-bounded such that $|f_1(x) - \frac{3}{4}x|^2 \leq \frac{1}{16} x^2$, with $f_1$ and its bounding region plotted in Fig. \ref{fig:example_bounding_regions}(\subref{fig:f1}). The bounding region is a cone in the linear-algebraic sense \cite[p. 278]{bernstein2018scalar}, motivating the use of the term cone-bounded. It can be verified from the definition \eqref{eq:cone_bounded_condition} that a cone-bounded function must be zero and Lipschitz continuous at the origin; we say that $f$ is Lipschitz at a point $x_0 \in \operatorname{domain}(f)$ if there exists an $L \geq 0$ such that $\|f(x) - f(x_0)\| \leq L \|x - x_0\|$ for every $x \in \operatorname{domain}(f)$. $f_2$ and $f_3$ demonstrate how quasi-cone-boundedness relaxes these restrictions.

$f_2$ cannot be cone-bounded, since $f_2$ is nonzero at the origin. However, $f_2$ is quasi-cone-bounded such that $|f_2(x) - f_2(0) - \frac{3}{4}x|^2 \leq \frac{1}{16} x^2$, which can be derived using the previous cone bound for $f_1$. Since $f_2(0)$ is nonzero, the bounding region is a \emph{shifted} cone (Fig. \ref{fig:example_bounding_regions}(\subref{fig:f2})).

Lastly, $f_3$ illustrates a benefit of the $\delta_t$ term in \eqref{eq:quasi_cone_bounded_condition}. While $f_3(0) = 0$, $f_3$ is not Lipschitz at the origin $x_0 = 0$, hence cannot be cone-bounded. Yet, it can be shown that $f_3$ is quasi-cone-bounded such that $|f_3(x)|^2 \leq a x^2 + b$ for some values of $a > 0$ and $\delta_t = b > 0$. Fig. \ref{fig:example_bounding_regions}(\subref{fig:f3}) depicts two corresponding bounding regions. In the example of $f_3$, it turns out that $b$ can never be zero; in the case of $b = 0$ and $a > 0$, there exists an $x$ such that $|f_3(x)|^2 > a x^2 $, e.g., $x = (2a)^{-1}$. Fig. \ref{fig:example_bounding_regions}(\subref{fig:f3}) exemplifies how the $\delta_t $ term can characterize a bounding region with a nonzero ``width'' at the origin and how such a region can bound a function like $f_3$ that behaves irregularly near the origin. Having motivated the extension to quasi-cone-boundedness, we consider the following assumption:

\begin{figure}[t]
    \centering
    \begin{subfigure}[t]{0.225\textwidth}
        \centering
        \includegraphics[width=\textwidth]{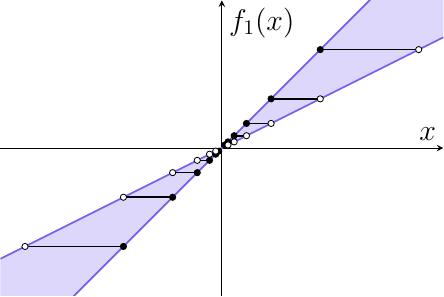}
        \caption{}
        \label{fig:f1}
    \end{subfigure}%
    \hfill%
    \begin{subfigure}[t]{0.225\textwidth}
        \centering
        \includegraphics[width=\textwidth]{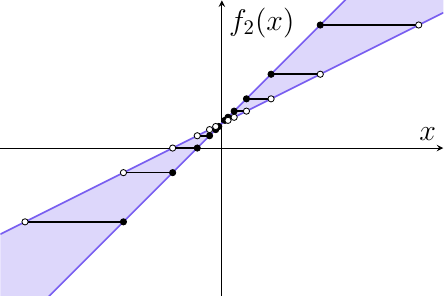}
        \caption{}
        \label{fig:f2}
    \end{subfigure}%
    \hfill%
    \begin{subfigure}[t]{0.225\textwidth}
        \centering
        \includegraphics[width=\textwidth]{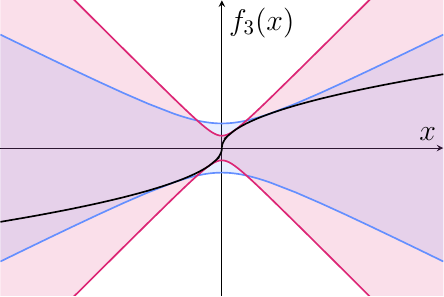}
        \caption{}
        \label{fig:f3}
    \end{subfigure}
    \caption{Example (quasi-)cone bounding regions for functions (a) $f_1(x) \triangleq \operatorname{sgn}(x) 2^{\lfloor\log_2|x|\rfloor}$ for $x \neq 0$, with $f_1(0) \triangleq 0$, (b) $f_2(x) \triangleq f_1(x) + 1$, and (c) $f_3(x) \triangleq \operatorname{sgn}(x) \sqrt{|x|}$.}
    \label{fig:example_bounding_regions}
\end{figure}

\begin{assum}[Quasi-cone-bounded nonlinearities]
    \label{assum:quasi-cone-bounded}
    For each $t \in \mathbb{T}$, $f_t$ in (\ref{eq:dynamics}) is quasi-cone-bounded.
\end{assum}

It can be verified that Assumption \ref{assum:quasi-cone-bounded} implies Assumption \ref{assum:dynamics}. Hence, we propose a well-posed optimal control problem with the stricter assumption of quasi-cone-boundedness replacing Assumption \ref{assum:dynamics}:
\begin{equation}
    \label{eq:main_optimal_control_problem}
    \begin{alignedat}{2}
        &\inf_\mathbf{u}         && \mathcal{J}_{\mathbf{u}} \\
        &\operatorname{s.t. } \; && x_{t + 1} = f_t(x_t, u_t) +  \tilde{A}_t x_t + \tilde{B}_t u_t + v_t \text{ and}\\
        &                        && \text{$u_t \in \mathcal{L}^2(\mathcal{F}_t)$ for each $t \in \mathbb{T}$; Assumptions \ref{assum:initial_state}, \ref{assum:quasi-cone-bounded}}.
    \end{alignedat}
\end{equation}
We denote the optimal value of \eqref{eq:main_optimal_control_problem} as $\mathcal{J}^*$. $\mathcal{J}_\mathbf{u}$ \eqref{eq:objective1} has risk-aware costs, which have not been considered in tandem with systems specifically having cone-bounded or quasi-cone-bounded nonlinearities to the authors' knowledge. After reformulating $\mathcal{J}_\mathbf{u}$ in Section \ref{sec:reformulating_variance_suppression}, we provide an analytical, suboptimal solution to \eqref{eq:main_optimal_control_problem} in Section \ref{sec:suboptimal_controller}.

\section{Reformulation of the Risk-Aware Objective}
\label{sec:reformulating_variance_suppression}

We specified the risk-aware cost $\mathbb{E}(\varrho_t)$ similarly to \cite{fujimoto2011optimalfinitetime}, but we cannot use their approach directly to tackle our problem since the system we are considering may be nonlinear. The authors of \cite{fujimoto2011optimalfinitetime} assume linear dynamics and feedback, and then define a map $\mathcal{A}_L : \mathbb{S}^n \to \mathbb{S}^n$ in terms of the dynamics and feedback matrices, which they use to reformulate the performance criterion \cite{fujimoto2011optimalfinitetime}. However, we neither assume linear dynamics nor linear feedback in \eqref{eq:main_optimal_control_problem}, and thus cannot similarly define $\mathcal{A}_L$. Fortunately, using elementary probability theory, we can reformulate $\varrho_t$ \eqref{eq:riskrandomvariable} into a convex quadratic expression in $(x_t,u_t)$ under Assumptions \ref{assum:initial_state}--\ref{assum:control}. This reformulation does \emph{not} require $f_t$ to be quasi-cone-bounded. We present a supporting lemma (Lemma \ref{lemma:support1}) and then the reformulation (Proposition \ref{prop:quadform}), in which $\xi_t \triangleq (x_t,u_t)$ and $\tilde{M}_t \triangleq [\tilde{A}_t \; \tilde{B}_t]$ for brevity. We also define $\Sigma_t \triangleq \text{cov}(v_t) = \mathbb{E}(v_t v_t^\top)$.
\begin{lem}
    \label{lemma:support1}
    Let the system \eqref{eq:dynamics} satisfy Assumptions \ref{assum:initial_state}--\ref{assum:control}. Then, for each $t \in \mathbb{T}$, $\mathbb{E}(x_{t+1} | \mathcal{F}_t) \overset{\mathrm{a.e.}}{=} f_t(\xi_t)$ and
    \begin{equation*}
        \begin{aligned}
            \mathbb{E}(x_{t+1}^\top P x_{t+1} | \mathcal{F}_t) & \overset{\mathrm{a.e.}}{=} \xi_t^\top \mathbb{E}(\tilde{M}_t^\top P\tilde{M}_t) \xi_t + \operatorname{tr}(P \Sigma_t) \\ & \hphantom{\overset{\mathrm{a.e.}}{=}} + f_t(\xi_t) ^\top P f_t(\xi_t)
        \end{aligned}
    \end{equation*}
    for any $P \in \mathbb{S}^n_{\geq}$.
\end{lem}
\begin{proof}
    \emph{Part 1:} Since $f_t(\xi_t)$, $\tilde{M}_t \xi_t$, and $v_t$ are integrable random vectors due to Assumptions \ref{assum:initial_state}--\ref{assum:control}, $\mathbb{E} (x_{t+1} |  \mathcal{F}_t ) \overset{\mathrm{a.e.}}{=} \mathbb{E}( f_t(\xi_t) | \mathcal{F}_t ) + \mathbb{E}(\tilde{M}_t \xi_t | \mathcal{F}_t) + \mathbb{E}(v_t | \mathcal{F}_t)$ \cite[Thm. 6.5.2(a)]{ash1972probability}. Since $\xi_t$ is $\mathcal{F}_t$-measurable and integrable, and since $\tilde{M}_t$ is independent of $h_t$ and zero-mean, $\mathbb{E}(\tilde{M}_t \xi_t | \mathcal{F}_t) \overset{\mathrm{a.e.}}{=} 0_{n}$ (Lemma \ref{lem:predictive_variance_helper3}). Also, $v_t$ is independent of $h_t$ and zero-mean, thus $\mathbb{E}(v_t | \mathcal{F}_t) \overset{\mathrm{a.e.}}{=} 0_{n }$. Lastly, $\mathbb{E}(f_t(\xi_t)|\mathcal{F}_t) \overset{\mathrm{a.e.}}{=} f_t(\xi_t)$ follows from $f_t(\xi_t)$ being $\mathcal{F}_t$-measurable and integrable. \emph{Part 2:} The second statement can be shown by expanding $x_{t+1}^\top P x_{t+1}$ into six terms, verifying that the terms are integrable, and calculating their conditional expectations. The derivation uses Assumptions \ref{assum:initial_state}--\ref{assum:control}, Lemmas \ref{lem:predictive_variance_helper3}--\ref{lem:predictive_variance_helper2}, and \cite[Thm. 6.5.2(a)]{ash1972probability}. For example, one term is $2 f_t(\xi_t)^\top P \tilde{M}_t \xi_t$, and $\mathbb{E}(2 f_t(\xi_t)^\top P \tilde{M}_t \xi_t|\mathcal{F}_t) \overset{\mathrm{a.e.}}{=} 0$ can be calculated using Lemma \ref{lem:predictive_variance_helper2}.
\end{proof}

\begin{prop}
    \label{prop:quadform}
    Let the system \eqref{eq:dynamics} satisfy Assumptions \ref{assum:initial_state}--\ref{assum:control}. Then, for each $t \in \mathbb{T}$, $\varrho_t$ \eqref{eq:riskrandomvariable} admits the following representation: $\varrho_t = \xi_t^\top \mathbb{E}(\tilde{M}_t^\top Z_t \tilde{M}_t) \xi_t + \operatorname{tr}(Z_t \Sigma_t) $ a.e.
\end{prop}
\begin{proof}
    $x_{t+1}$ is square-integrable under Assumptions \ref{assum:initial_state}--\ref{assum:control}, $\mathcal{F}_t$ is a sub-$\sigma$-algebra of $\mathcal{F}$, and $Z_t \in \mathbb{S}^n_\geq$, thus we have $\varrho_t\stackrel{\mathclap{\normalfont\mbox{\tiny a.e.}}}{=} \mathbb{E}(x_{t+1}^\top Z_t x_{t+1} |\mathcal{F}_t) - \mathbb{E}(x_{t+1} | \mathcal{F}_t)^\top Z_t\mathbb{E}(x_{t+1} | \mathcal{F}_t)$ (Lemma \ref{lem:predictive_variance_helper1}). This expression appears in \cite[p. 12606]{fujimoto2011optimalfinitetime} but for a linear system and with conditioning on $x_t$; we consider a nonlinear system and condition on $h_t$ instead. Then, substituting the expressions for $\mathbb{E}(x_{t+1}^\top Z_t x_{t+1} |\mathcal{F}_t)$ and $\mathbb{E}(x_{t+1} | \mathcal{F}_t)$ from Lemma \ref{lemma:support1} yields the reformulation of $\varrho_t$ in Proposition \ref{prop:quadform}.
\end{proof}

Proposition \ref{prop:quadform} simplifies the risk-aware optimal control problem \eqref{eq:main_optimal_control_problem}, since it allows us to rewrite the objective under Assumptions \ref{assum:initial_state}--\ref{assum:control} as $\mathcal{J}_{\mathbf{u}} = \mathbb{E}(c_T(x_T) + \sum_{t=0}^{T-1} c_t^Z(x_t,u_t)) \in \mathbb{R}$, where $c_t^Z(\xi) \triangleq \xi^\top G_t^Z \xi + \operatorname{tr}(Z_t \Sigma_t)$ for every $\xi = (x,u) \in \mathbb{R}^n \times \mathbb{R}^m$ and
\begin{equation}
    \label{eq:Gtz}
    G_t^Z \triangleq \genfrac{[}{]}{0pt}{0}{Q_t \;\;\, S_t}{S_t^\top \;\: R_t} + \mathbb{E}\bigl(\bigl[\tilde{A}_t \; \tilde{B}_t\bigr]^\top Z_t \bigl[\tilde{A}_t \; \tilde{B}_t\bigr]\bigr)
\end{equation}
for every $t \in \mathbb{T}$. Hence, multiplicative noise, but not additive noise, will play a role in the proposed controller (Section \ref{sec:suboptimal_controller}). Since $\mathcal{J}_{\mathbf{u}}$ equals the expected sum of some (convex) quadratic costs, the problem \eqref{eq:main_optimal_control_problem} reduces to a risk-neutral alternative obtained by modifying the matrices $Q_t$, $R_t$, and $S_t$ of the original cost function \eqref{eq:costs}. Tackling risk-aware problems by reducing them to risk-neutral ones is an oft-used technique \cite[Sect. 4.2.3]{wang2022risk}, \cite{tsiamis2020riskconstrained}, \cite{weitobesubmitted}. The equivalence of \eqref{eq:main_optimal_control_problem} to a risk-neutral alternative does \emph{not} negate the utility of the original risk-aware formulation \cite{majumdar2019how}. Merely solving a standard risk-neutral problem at the start begs the question of how to compensate systematically for the increased variability in the state arising from the multiplicative disturbances. Fortunately, the problem \eqref{eq:main_optimal_control_problem} penalizes this variability by design via the variance-like cost $\mathbb{E}(\varrho_t)$.

\section{Analytical, Suboptimal Controller}
\label{sec:suboptimal_controller}

Having reformulated $\mathcal{J}_\mathbf{u}$, we offer an analytical, suboptimal solution to \eqref{eq:main_optimal_control_problem}, resembling the form of LQR controllers, enabled by the quasi-cone-bounded structure of the nonlinearities. We purposefully determine a suboptimal solution (with desirable attributes, see below) rather than an optimal one. The optimal control of a general nonlinear system is challenging due to lack of convenient structure. Taking inspiration from \emph{Anna Karenina}, ``all linear systems are alike; every type of nonlinear system is nonlinear in its own way'' \cite{tolstoy2014anna}. A key issue is that the nonlinear dynamics \eqref{eq:dynamics} impose equality constraints that render (\ref{eq:main_optimal_control_problem}) nonconvex, precluding the use of convex optimization to compute an optimal controller. Furthermore, applying dynamic programming to compute an optimal controller can run into Bellman's ``curse of dimensionality'' \cite{bertsekas1996neuro}. A common technique to alleviate this curse is to approximate the optimal cost-to-go functions using a family of parameterized functions, a method of approximate dynamic programming \cite{bertsekas1996neuro,bertsekas2019reinforcement}. Using a quadratic parameterization of cost-to-go random variables, we analyze our proposed suboptimal solution $\hat{\mathbf{u}}$ to \eqref{eq:main_optimal_control_problem} in Theorem \ref{thm:regulation_upper_bounds}. A key feature of $\hat{\mathbf{u}}$ is its explicit dependence on the quasi-cone parameters of $f_t$ in the bound \eqref{eq:quasi_cone_bounded_condition}, instead of depending on $f_t$ exactly (Theorem \ref{thm:regulation_upper_bounds}). Considering a class of nonlinearities obeying a particular bound, instead of one specific nonlinearity, is standard in classical nonlinear control, e.g., in stability analysis concerning sector-bounded nonlinearities \cite[p. 264]{khalil2002nonlinear}. A similar concept is that of relaxation in convex optimization, where an equality constraint (analogously, a precise form of $f_t$) is relaxed to an inequality constraint (analogously, a bound on $f_t$) \cite[p. 654]{boyd2004convex}. Advantages of the suboptimal solution $\hat{\mathbf{u}}$ include its ease of computation and familiar analytical form, the latter enhancing its interpretability.

Before presenting $\hat{\mathbf{u}}$ in Theorem \ref{thm:regulation_upper_bounds}, we present an intermediary result: a quadratic upper bound of a quadratic form of a quasi-cone-bounded function. This result adapts the matrix bounds from \cite[p. 51]{scherzinger1982estimation} using a generalized Peter--Paul inequality to enable the analysis of $\hat{\mathbf{u}}$ in Theorem \ref{thm:regulation_upper_bounds}. The result uses $\operatorname{tr}(\cdot)$ or $\|\cdot\|$ to measure the size of a matrix $P \in \mathbb{S}_\geq^n$. We use $\operatorname{tr}(\cdot)$ in an example with $n = 1000$ due to its computational advantage (Section \ref{sec:numerical_results}).

\begin{lem}
    \label{lem:quadratic_bound}
    Let $f : \mathbb{R}^n \times \mathbb{R}^m \rightarrow \mathbb{R}^n$ be quasi-cone-bounded, with parameters $f^0 \triangleq f(0_n,0_m)$, $A \in \mathbb{R}^{n\times n}$, $B \in \mathbb{R}^{n \times m}$, $\Delta \in \mathbb{S}_\geq^{n+m}$, and $\delta \geq 0$ in \eqref{eq:quasi_cone_bounded_condition}. Given $P \in \mathbb{S}_\geq^n$ and $q \in \mathbb{R}^n$, define $\psi(x) \triangleq x^\top P x + 2 x^\top q$ on $\mathbb{R}^n$. Let $\varphi(\cdot)$ be either $\operatorname{tr}(\cdot)$ or $\|\cdot\|$. Given $\alpha > 0$ and $\beta > 0$, define $ \lambda \triangleq (1 + \alpha^{-1}) \varphi(P) + \beta^{-1}$, \allowdisplaybreaks
    \begin{align*}
        \begin{split}
            \bar{P} &\triangleq (1 + \alpha) [A \;\: B]^\top P [A \;\: B] + \lambda \Delta,
        \end{split} \\
            \bar{q}
        &\triangleq [A \;\: B]^\top ((1 + \alpha) P f^0 + q),  \quad  \text{and} \\
        \begin{split}
            \bar{r} &\triangleq (1 + \alpha) f^{0\top} P f^0 + \beta q^\top q + 2  q^\top f^{0} + \lambda \delta .
        \end{split}
    \end{align*}
    Then, $\psi(f(\xi)) \leq \xi^\top \bar{P} \xi + 2 \xi^\top \bar{q} + \bar{r}$ for every $\xi \triangleq (x,u) \in \mathbb{R}^n\times\mathbb{R}^m$, with $\bar{P} \in \mathbb{S}_\geq^{n+m}$.
\end{lem}
\begin{proof}
    \allowdisplaybreaks
    Denote $M \triangleq [A \; B]$ and $\varepsilon \triangleq f(\xi) - (f^0 + M \xi )$, so \eqref{eq:quasi_cone_bounded_condition} can be written as $ \|\varepsilon\|^2 \leq \xi^\top \Delta \xi + \delta$. Substituting $f(\xi) = \varepsilon + f^0 + M \xi $ into $\psi(f(\xi))$ yields $\psi(f(\xi)) = \text{Term1} + \text{Term2} + 2(f^0 + M \xi )^\top q$, where $\text{Term1} \triangleq  (\varepsilon + f^0 + M \xi )^\top P(\varepsilon + f^0 + M \xi )$ and $\text{Term2} \triangleq 2 \varepsilon^\top q$ depend on $f$ via $\varepsilon$. So, we seek to bound Term1 and Term2. Applying Lemma \ref{lem:peter_paul_inequality}, we derive $\text{Term1}  \leq (1 + \alpha^{-1}) \varepsilon^\top P \varepsilon  + (1 + \alpha)(*)^\top P (f^0 + M \xi )$ and $\text{Term2} \leq \beta^{-1} \varepsilon^\top\varepsilon + \beta q ^\top q$, respectively. Note that $\varepsilon^\top P \varepsilon \leq \varphi(P) \|\varepsilon\|^2$. Applying the above inequalities to the expansion of $\psi(f(\xi))$ and grouping like terms leads to $\psi(f(\xi)) \leq \xi^\top \bar{P} \xi + 2 \xi^\top \bar{q} + \bar{r}$.
\end{proof}

An optimal controller in a standard, stochastic finite-horizon LQR setting enjoys an analytical form, whereas the analytical feature is typically lost when the dynamics of interest are nonlinear \cite{hernandezlerma2012discrete}. Fortunately, when each $f_t$ in the dynamics \eqref{eq:dynamics} is \emph{quasi-cone-bounded}, we offer a controller $\hat{\mathbf{u}}$ that has an analytical form and a conservative performance guarantee with respect to \eqref{eq:main_optimal_control_problem}, presented next. For convenience, we use the notation
\begin{equation*}
    \Delta_t = \genfrac{[}{]}{0pt}{0}{\hphantom{(}\Delta_t^{xx}\hphantom{)^\top} \;\: \Delta_t^{xu}}{(\Delta_t^{xu})^\top \;\: \Delta_t^{uu}},
\end{equation*}
where $\Delta_t^{xx} \in \mathbb{S}^n_\geq$ and $\Delta_t^{uu} \in \mathbb{S}^m_\geq$ due to $\Delta_t \in \mathbb{S}^{n+m}_\geq$.

\begin{thm}
    \label{thm:regulation_upper_bounds}
    Under Assumptions \ref{assum:initial_state} and \ref{assum:quasi-cone-bounded}, let the system \eqref{eq:dynamics} operate under a controller $\hat{\mathbf{u}} \triangleq (\hat{u}_0,\dots,\hat{u}_{T-1})$ to be specified, and denote the corresponding states as $\hat{x}_0,\hat{x}_1,\dots,\hat{x}_T$, with $\hat{x}_0 \triangleq x_0$. For each $t \in \mathbb{T}$, let the parameters $\alpha_t > 0$ and $\beta_t > 0$ be given, and define $\hat{u}_t$ as $\hat{u}_t \triangleq K_t \hat{x}_t + \ell_t$, where $K_t$ and $\ell_t$ are calculated via the following Riccati-like recursion: \allowdisplaybreaks
    \begin{align}
            K_t
        &\triangleq  -\bigl(\hat{R}_t + \hat{B}_t^\top P_{t + 1} \hat{B}_t\bigr)^{-1} \bigl(\hat{S}_t + \hat{A}_t^\top P_{t + 1} \hat{B}_t\bigr)^\top, \label{Kt}\\
            \ell_t
        &\triangleq  -\bigl(\hat{R}_t + \hat{B}_t^\top P_{t + 1} \hat{B}_t\bigr)^{-1} B_t^\top \hat{f}_t^0, \label{eq:lhat}\\
            \hat{A}_t
        &\triangleq  (1 + \alpha_t)^{\nicefrac{1}{2}} A_t, \label{eq:Ahat} \\
            \hat{B}_t
        &\triangleq  (1 + \alpha_t)^{\nicefrac{1}{2}} B_t, \label{eq:Bhat} \\
            \hat{f}_t^0
        &\triangleq  (1 + \alpha_t) P_{t + 1} f_t^0 + q_{t + 1}, \\
            \hat{Q}_t
        &\triangleq  Q_t + \mathbb{E}(\tilde{A}_t^\top (P_{t + 1} + Z_t) \tilde{A}_t) + \lambda_t \Delta_t^{xx}, \label{eq:Qhat} \\
            \hat{R}_t
        &\triangleq  R_t + \mathbb{E}(\tilde{B}_t^\top (P_{t + 1} + Z_t) \tilde{B}_t) + \lambda_t \Delta_t^{uu}, \label{eq:Rhat} \\
            \hat{S}_t
        &\triangleq  S_t + \mathbb{E}(\tilde{A}_t^\top (P_{t + 1} + Z_t) \tilde{B}_t) + \lambda_t \Delta_t^{xu}, \label{eq:Shat} \\
        \begin{split}
                P_t
            &\triangleq  \hat{Q}_t + \hat{A}_t^\top P_{t + 1} \hat{A}_t \!-\! K_t^\top (\hat{R}_t + \hat{B}_t^\top P_{t + 1} \hat{B}_t) K_t, \label{Pt}
        \end{split} \\
            q_t
        &\triangleq  (A_t + B_t K_t)^\top ((1 + \alpha_t) P_{t + 1} f_t^0 + q_{t + 1}), \label{qt}\\
        \begin{split}
                r_t
            &\triangleq  r_{t + 1} + \operatorname{tr}((P_{t + 1} + Z_t) \Sigma_t) + \beta_t q_{t + 1}^\top q_{t + 1} \\ &\mathrel{\phantom{=}} + 2 q_{t + 1}^\top f_t^0 + (1 + \alpha_t) f_t^{0\top} P_{t + 1} f_t^0 \\ &\mathrel{\phantom{=}} - \ell_t^\top \bigl(\hat{R}_t + \hat{B}_t^\top P_{t + 1} \hat{B}_t\bigr) \ell_t + \lambda_t \delta_t, \quad \text{and} \label{rt}
        \end{split} \\
            \lambda_t
        &\triangleq   (1 + \alpha_t^{-1}) \varphi(P_{t + 1}) + \beta_t^{-1} \label{eq:lambdat},
    \end{align}
    with $P_T \triangleq Q_T$, $q_T \triangleq 0_n$, $r_T \triangleq 0$, and $\varphi(\cdot)$ being either $\operatorname{tr}(\cdot)$ or $\|\cdot\|$. Then, $\hat{\mathbf{u}}$ is a suboptimal solution to \eqref{eq:main_optimal_control_problem}, with an upper bound $\mathcal{J}_{\hat{\mathbf{u}}} \leq \mathbb{E}(x_0^\top P_0 x_0) + 2 \mathbb{E}(x_0)^\top q_0 + r_0$.
\end{thm}
\begin{proof}
    For each $t \in \mathbb{T}$, let $\hat{h}_t \triangleq (\hat{x}_0,\hat{u}_0,\hat{x}_1,\dots,\hat{u}_{t-1},\hat{x}_t)$ and $\hat{\mathcal{F}}_t \triangleq \sigma(\hat{h}_t)$, with $\hat{h}_0 \triangleq \hat{x}_0$. The assumed conditions and definition of $\hat{\mathbf{u}}$ ensure that $\hat{\xi}_t \triangleq (\hat{x}_t,\hat{u}_t)$ and $\hat{x}_T$ are square-integrable, and $\hat{u}_t$ is $\hat{\mathcal{F}}_t$-measurable. To upper bound $\mathcal{J}_{\hat{\mathbf{u}}}$, first define $V_T \triangleq Y_T \triangleq c_T(\hat{x}_T)$ and $V_t \triangleq \mathbb{E}(Y_t|\hat{\mathcal{F}}_t)$, with $Y_t \triangleq c_T(\hat{x}_T) + \sum_{\tau = t}^{T-1} c_\tau^Z(\hat{\xi}_\tau)$ for each $t \in \mathbb{T}$. Given $P \in \mathbb{R}^{n \times n}$, $q \in \mathbb{R}^n$, and $r \in \mathbb{R}$, let $\phi(x;P,q,r) \triangleq x^\top P x + 2 x^\top q + r$, $x \in \mathbb{R}^n$. It suffices to show that $V_t \leq \phi(\hat{x}_t; P_t,q_t,r_t)$ a.e. for each $t \in \{0,\dots,T\}$. Proceed by induction. \emph{Base case:} $V_T = \hat{x}_T^{ \top} Q_T \hat{x}_T = \phi(\hat{x}_T; P_T, q_T, r_T)$, with $P_T = Q_T \in \mathbb{S}^n_\geq$. \emph{Induction hypothesis:} Assume that, for some $t \in \mathbb{T}$, there exist $P_{t+1} \in \mathbb{S}_\geq^n$, $q_{t+1} \in \mathbb{R}^n$, and $r_{t+1} \in \mathbb{R}$ s.t. $V_{t+1} \leq \phi(\hat{x}_{t+1}; P_{t+1}, q_{t+1}, r_{t+1})$ a.e. \emph{Induction step:} We use the induction hypothesis and $V_t = \mathbb{E}(V_{t+1}|\hat{\mathcal{F}}_t) + c_t^Z(\hat{\xi}_t)$ a.e. to derive $V_t \leq c_t^Z(\hat{\xi}_t) + \Phi_t + r_{t+1}$ a.e., where $\Phi_t \triangleq \mathbb{E}(\hat{x}_{t+1}^{\top} P_{t+1} \hat{x}_{t+1} |\hat{\mathcal{F}}_t) + 2 \mathbb{E}(\hat{x}_{t+1} | \hat{\mathcal{F}}_t)^\top q_{t+1}$. We apply Lemma \ref{lemma:support1} and $P_{t+1} \in \mathbb{S}_\geq^{n}$ to deduce that $\Phi_t =\hat{\xi}_t^{\top} \mathbb{E}(\tilde{M}_t^\top P_{t+1} \tilde{M}_t) \hat{\xi}_t + \operatorname{tr}(P_{t+1}\Sigma_t) + \Phi_t'$ a.e., with $\tilde{M}_t \triangleq [\tilde{A}_t \; \tilde{B}_t]$ and $\Phi_t' \triangleq f_t(\hat{\xi}_t)^\top P_{t+1} f_t(\hat{\xi}_t) + 2 f_t(\hat{\xi}_t)^\top q_{t+1}$. Next, we use Assumption \ref{assum:quasi-cone-bounded} and Lemma \ref{lem:quadratic_bound} to obtain $\Phi_t' \leq \hat{\xi}_t^{\top} \bar{P}_t \hat{\xi}_t + 2 \hat{\xi}_t^{\top} \bar{q}_t + \bar{r}_t$, where $\bar{P}_t \in \mathbb{S}^{n+m}_\geq$, $\bar{q}_t$, and $\bar{r}_t$ depend on the quasi-cone parameters. All together,
    \begin{equation}\label{eq:almostthere}
        V_t \overset{\mathrm{a.e.}}{\leq}\hat{\xi}_t^{\top} L_t \hat{\xi}_t + 2\hat{\xi}_t^{\top} \bar{q}_t + r_{t+1} + \operatorname{tr}((P_{t+1} + Z_t)\Sigma_t) + \bar{r}_t,
    \end{equation}
    where $L_t \triangleq G_t^Z + \mathbb{E}(\tilde{M}_t^\top P_{t+1} \tilde{M}_t) + \bar{P}_t $ can be written as 
       \begin{equation*}
        L_t = \begin{bmatrix} \hat{Q}_t + \hat{A}_t^\top P_{t+1} \hat{A}_t & \hat{S}_t + \hat{A}_t^\top P_{t+1} \hat{B}_t \\ (\hat{S}_t + \hat{A}_t^\top P_{t+1} \hat{B}_t)^\top & \hat{R}_t + \hat{B}_t^\top P_{t+1} \hat{B}_t  \end{bmatrix} ,
    \end{equation*}
    with $\hat{R}_t + \hat{B}_t^\top P_{t+1} \hat{B}_t \in \mathbb{S}_>^m$.
    Lastly, we apply quadratic optimization (Lemma \ref{lemma:opt}) and the definitions \eqref{Pt}--\eqref{rt} to deduce that the right side of \eqref{eq:almostthere} equals $\hat{x}_t^{\top} P_t \hat{x}_t + 2 \hat{x}_t^{\top} q_t + r_t$. $P_t$ can be written as $P_t = \kappa_t^\top L_t \kappa_t$, where $\kappa_t \triangleq [I_n \;\; K_t^\top]^\top$ and $L_t \in \mathbb{S}^{n+m}_\geq$, so $P_t \in \mathbb{S}_\geq^n$.
\end{proof}

The formulae for $K_t$ and $P_t$ in Theorem \ref{thm:regulation_upper_bounds} resemble those from a standard Riccati recursion for stochastic LQR. However, the usual matrices $Q_t$ and $R_t$ in the cost function \eqref{eq:costs} are inflated by terms depending on the risk parameter $Z_t$, the multiplicative noise, and the quasi-cone parameter $\Delta_t$ (see \eqref{eq:Qhat}--\eqref{eq:Rhat}). The quasi-cone parameters $A_t$ and $B_t$ that characterize the linear part of the nominal model of $f_t$ are inflated due to a design parameter $\alpha_t$ (see \eqref{eq:Ahat}--\eqref{eq:Bhat}). To reason about how to select $\alpha_t$, consider the cone-bounded case and rewrite the bound on the cost-to-go $V_t$ as
\begin{align}
    V_t &\leq \hat{\xi}_t^\top \bigl(\alpha_t [A_t \; \; B_t]^\top P_{t + 1} [A_t \; \; B_t] + \alpha_t^{-1} \varphi(P_{t + 1}) \Delta_t\bigr) \hat{\xi}_t \nonumber \\ &\phantom{{}={}} + [\text{terms not depending on $\alpha_t$}] \text{ a.e.} \nonumber \\
    &\leq \bigl(\alpha_t \|[A_t \; \; B_t]\|^2 + \alpha_t^{-1} \|\Delta_t\|\bigr) \varphi(P_{t + 1}) \|\hat{\xi}_t\|^2 \label{eq:approx_bound}\\ &\phantom{{}={}} + [\text{terms not depending on $\alpha_t$}]. \nonumber
\end{align}
The terms depending on $\alpha_t$ arise as a trade-off when applying a Peter--Paul inequality. $\alpha_t$ allows the user to tune the weights on these terms, with a smaller choice of $\alpha_t$ generally better when the ``magnitude'' of the nonlinearity, $\|\Delta_t\|$, is small compared to $\|[A_t \; \; B_t]\|^2$. More precisely, \eqref{eq:approx_bound} with some elementary calculus suggests an ``optimal'' $\alpha_t$ of about $\|\Delta_t\|^{\nicefrac{1}{2}} \|[A_t \; \; B_t]\|^{-1}$. In the quasi-cone-bounded case, we may perform similar analysis demonstrating that a smaller choice of $\beta_t$ is generally better when $\|\Delta_t\|$ and $\delta_t$ are small compared to $\|f_t^0\|^2$, however it is more difficult to estimate an ``optimal'' $\beta_t$.

Lastly, we consider a special case of the gains from Theorem \ref{thm:regulation_upper_bounds}. We find that such gains characterize an optimal solution to a problem related to \eqref{eq:main_optimal_control_problem}, but with $f_t$ set to an inflated linear model.

\begin{lem}
    \label{lemma12}
    For each $t \in \mathbb{T}$, given some $\alpha_t > 0$, $A_t \in \mathbb{R}^{n \times n}$, and $B_t \in \mathbb{R}^{n \times m}$, define $\hat{A}_t$ and $\hat{B}_t$ as in \eqref{eq:Ahat} and \eqref{eq:Bhat}, respectively. Furthermore, specify $K_t$, $P_t$, $q_t$, and $r_t$ as in \eqref{Kt}, \eqref{Pt}, \eqref{qt}, and \eqref{rt}, respectively, but in the special case of $f_t^0 = 0_n$, $\Delta_t = 0_{n\times n}$, and $\delta_t = 0$ for each $t \in \mathbb{T}$; use $P_T = Q_T$, $q_T = 0_n$, and $r_T = 0$. Consider
    \begin{equation}
    \label{eq:simple_optimal_control_problem}
        \begin{alignedat}{2}
            &\inf_\mathbf{u}         && \mathcal{J}_{\mathbf{u}} \\
            &\operatorname{s.t. } \; && x_{t + 1} = (\hat{A}_t + \tilde{A}_t) x_t    +  (\hat{B}_t + \tilde{B}_t) u_t + v_t  \text{ and} \\
            &                        && \text{$u_t \in \mathcal{L}^2(\mathcal{F}_t)$ for each $t \in \mathbb{T}$; Assumption \ref{assum:initial_state}}.
        \end{alignedat}
    \end{equation}
    Define $\mathbf{u}^* \triangleq (u_0^*,\dots,u_{T-1}^*)$, with $u_t^* \triangleq K_t x_t^*$ for each $t \in \mathbb{T}$; $x_0^*,x_1^*,\dots,x_T^*$ is the state trajectory satisfying the dynamics in \eqref{eq:simple_optimal_control_problem} under $\mathbf{u}^*$, with $x_0^* \triangleq x_0$. Denote the optimal value of the risk-aware problem \eqref{eq:simple_optimal_control_problem} as $\mathcal{V}^*$. Then, $\mathcal{V}^* = \mathcal{J}_{\mathbf{u}^*} = \mathbb{E}(x_0^\top P_0 x_0) + r_0$.
\end{lem}
\begin{proof}
    Under the specified conditions, $q_t$ and $r_t$ simplify to $q_t = 0_n$ and $r_t = r_{t+1} + \operatorname{tr}((P_{t+1} + Z_t)\Sigma_t)$, respectively, for each $t \in \mathbb{T}$. $\hat{Q}_t$ \eqref{eq:Qhat}, $\hat{R}_t$ \eqref{eq:Rhat}, and $\hat{S}_t$ \eqref{eq:Shat} also simplify due to $\Delta_t = 0_{n \times n}$; e.g., $\hat{Q}_t = Q_t + \mathbb{E}(\tilde{A}_t^\top(P_{t+1} + Z_t) \tilde{A}_t)$. Since Assumptions \ref{assum:initial_state}--\ref{assum:control} hold in \eqref{eq:simple_optimal_control_problem}, we can rewrite the objective as  $\mathcal{J}_{\mathbf{u}} = \mathbb{E}(c_T(x_T) + \sum_{t=0}^{T-1} c_t^Z(x_t,u_t))$, as discussed near \eqref{eq:Gtz}. Then, apply dynamic programming \cite[Thm. 2.15, p. 75]{kumar1986stochastic} and quadratic optimization (Lemma \ref{lemma:opt}).
\end{proof}

While differences arise in presentation, $\mathbf{u}^*$ is equivalent to the optimal controller from  \cite{fujimoto2011optimalfinitetime} with minor modifications (e.g., set $\hat{A}_t + \tilde{A}_t$ to their stochastic dynamics matrix and $S_t$ to zero, and use a time-invariant setting).

\section{Numerical Examples}
\label{sec:numerical_results}

Returning to the controller from Theorem \ref{thm:regulation_upper_bounds}, first we numerically evaluate its performance in the cone-bounded case for ease of comparison to earlier work \cite{scherzinger1982estimation}, which did not study the quasi-cone-bounded case. We consider the quasi-cone-bounded case in a high-dimensional setting at the end of this section. 

Define a quantizer $\psi_\gamma$ for $\gamma > 1$ as $\psi_\gamma(0) \triangleq 0$ and $\psi_\gamma(z) \triangleq \frac{2}{\gamma + 1} \operatorname{sgn}(z) \gamma^{\lceil \log_\gamma |z| \rceil}$ for $z \ne 0$ (similar to Fig. \ref{fig:example_bounding_regions}(\subref{fig:f1})). Also define $\psi_1(z) \triangleq z$. With $x_0 \sim \mathcal{N}(0_2, 4 I_2)$ and $\tilde{w}_t \sim \mathcal{N}(0, 0.25)$, consider the dynamics
\begin{equation}
    \label{eq:dummy_system}
    x_{t + 1} =
    \begin{bmatrix}
        a_1^\top x_t + \tilde{w}_t (x_t)_2 \\
        \psi_\gamma(a_2^\top x_t + u_t)
    \end{bmatrix}
\end{equation}
over a horizon of length $T = 10$, where $a_1 = (1, 0.75)$ and $a_2 = (0.75, -1)$. The system is cone-bounded with cone parameters $A_t = [a_1 \ a_2]^\top$, $B_t = (0, 1)$, and $\Delta_t = \bigl(\frac{\gamma - 1}{\gamma + 1}\bigr)^2 [a_2^\top \ 1]^\top [a_2^\top \ 1]$. This implies that the system is linear when $\gamma = 1$, and becomes more nonlinear as $\gamma$ increases. It is also cone-bounded in the sense of \cite{scherzinger1982estimation} with cone parameters $a = 1.25 \bigl(\frac{\gamma - 1}{\gamma + 1}\bigr)$ and $b = \frac{\gamma - 1}{\gamma + 1}$. We choose $Q_t = I_2$, $R_t = 1$, and $S_t = 0_{2 \times 1}$, and estimate $\mathcal{J}_\mathbf{u}$ for each controller considered as the mean cost incurred with $N = 10^6$ randomly generated disturbance processes. We evaluate controller performance first by varying the degree of nonlinearity, then by varying the risk-awareness.

First, we consider the risk-neutral case $Z_t = 0_{2 \times 2}$ and determine the controllers in Theorem \ref{thm:regulation_upper_bounds} with $\alpha_t \in \{\nicefrac{1}{2}, 1, 2\}$, denoted as $\hat{\mathbf{u}}^{(\nicefrac{1}{2})}$, $\hat{\mathbf{u}}^{(1)}$, and $\hat{\mathbf{u}}^{(2)}$, respectively. Fig. \ref{fig:varying_gamma} (top) shows the value of $\mathcal{J}_\mathbf{u}$ for each controller as $\gamma$ varies. From this plot, we observe that the better choice of $\alpha_t$ is generally larger when $\gamma$ is larger, which is in line with our discussion after Theorem \ref{thm:regulation_upper_bounds}. In that discussion, we suggested an ``optimal'' choice of $\alpha_t$ of about $\|\Delta_t\|^{\nicefrac{1}{2}} \|[A_t \; \; B_t]\|^{-1}$, which equals $\frac{\gamma - 1}{\gamma + 1}$ in this example. So, we compare our suboptimal controller with this choice of $\alpha_t$, denoted $\hat{\mathbf{u}}^{\alpha^*}$, to the one proposed in \cite[Thm. 4, eq. (25)]{scherzinger1982estimation}, $\mathbf{u}^{\text{orig}}$, in Fig. \ref{fig:varying_gamma} (bottom). Both controllers perform identically when $\gamma = 1$, i.e., when the system is linear, however our proposed controller retains better performance for larger $\gamma$.

\begin{figure}[htbp]
    \centering
    \includegraphics[width=\linewidth]{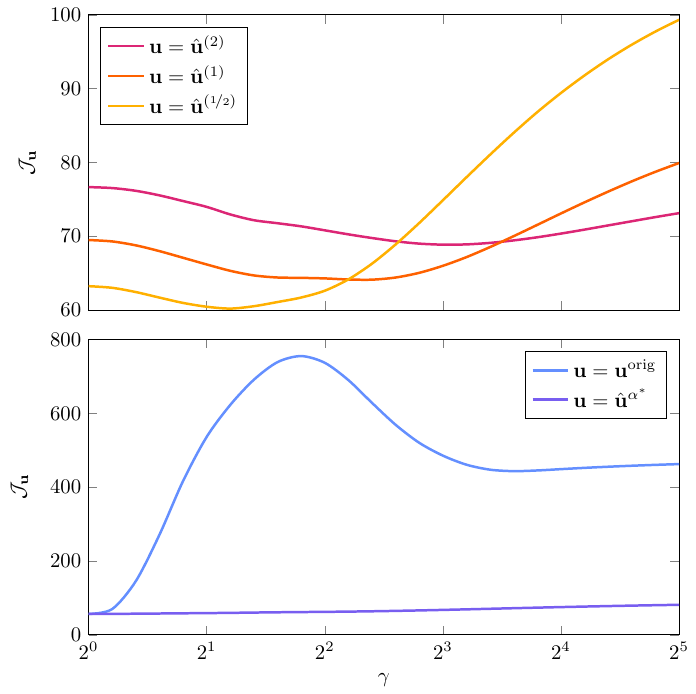}
    \caption{For different choices of $\mathbf{u}$, we plot $\mathcal{J}_\mathbf{u}$ \eqref{eq:objective1}, with $Z_t = 0_{2 \times 2}$, incurred by the system (\ref{eq:dummy_system}) for various $\gamma$.}
    \label{fig:varying_gamma}
\end{figure}

Now, we fix $\gamma = 10$, then evaluate risk-awareness by choosing nonzero $Z_t = \operatorname{diag}(z, 0)$, for $z > 0$. This risk-aware cost aims to mitigate the expected conditional variance of the first state. As we see in Fig. \ref{fig:varying_z}, this is achieved by our proposed controller with $\hat{\mathbf{u}}^{(2)}$ even when the system is nonlinear, with the risk-aware cost being decreased by almost $50\%$ for appropriately large $z$. Somewhat counterintuitively, this reduction in variability of the first state is accomplished by penalizing the \emph{second} state, since $\varrho_t = 0.25 z (x_t)_2^2$ a.e. A naive, and less effective, approach would have simply introduced an additional cost on the first state. In our example, one may deduce that the second state should be penalized to reduce the variability of the first state, since the disturbance is multiplied by the second state. However, in more complicated systems with various sources of multiplicative disturbances, appropriate choices of additional costs may be highly unintuitive, in which case a variance-suppression-based approach would provide a systematic method of reducing variability in the system.

\begin{figure}[htbp]
    \centering
    \includegraphics[width=\linewidth]{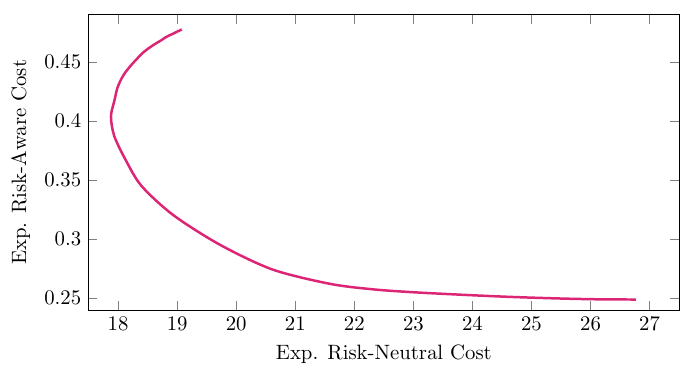}
    \caption{The total expected conditional variance $\mathbb{E}\bigl(\sum_{t = 1}^T \operatorname{cov}((x_t)_1 | \mathcal{F}_t)\bigr)$ plotted against the expected risk-neutral cost incurred by our controller with $\alpha_t = 2$ as $z$ varies from $10^{-3}$ to $10^6$. The arrows indicate the direction of increasing $z$, i.e., increasing risk-awareness.}
    \label{fig:varying_z}
\end{figure}

Previously, we presented an example with a cone-bounded nonlinearity. Now, we evaluate the viability of the quasi-cone-bounded controller from observing that a simple MATLAB script can implement it for $n = 1000$ states and $m = 100$ controls over a horizon of length $T = 10$ in just 1.4 seconds (on a laptop with an i7-1165G7 @ 2.80GHz processor). For general nonlinear systems, a numerical dynamic programming approach would be unviable for a stochastic system of such size.

\section{Conclusion}
\label{sec:conclusion}

This article develops a tractable, principled approach to risk-aware control for a class of nonlinear systems. While this contribution helps bridge the gap between risk-aware control and complex systems, much more work is needed. One extension is to develop a risk-aware controller that accounts for stochastic variability arising from both multiplicative and additive noise. An initial idea in this direction is to combine the risk assessments from \cite{fujimoto2011optimalfinitetime} and \cite{tsiamis2020riskconstrained}. Additional directions include considering full-distribution uncertainty characterizations \cite{wang2025policy}, partial observability, and multi-agent systems \cite{patel2024riskaware}.

\section*{Acknowledgment}
\label{sec:ack}

M.P.C. thanks Marco Pavone, Huizhen Yu, and Dionysis Kalogerias for helpful advice, and an anonymous reviewer of \cite{chapman2025riskaware} for pointing out the work of \cite{fujimoto2011optimalfinitetime}.

\section*{Appendix}

We include facts below to help this work be more self-contained. These facts or variations of them appear or are used elsewhere \cite{boyd2004convex,fujimoto2011optimalfinitetime,tsiamis2020riskconstrained,chapman2025riskaware,patel2025risk,weitobesubmitted}. A variation of Lemma \ref{lem:predictive_variance_helper3} can be found in \cite[Prop. 1(a)]{weitobesubmitted}. Lemma \ref{lem:predictive_variance_helper2} appears in \cite[Lemma 4.2]{patel2025risk} and \cite[Prop. 2]{weitobesubmitted}. A special case of the result of Lemma \ref{lem:predictive_variance_helper1} appears in \cite[p. 12606]{fujimoto2011optimalfinitetime}. Lemmas \ref{lem:predictive_variance_helper3}--\ref{lem:predictive_variance_helper2} follow from elementary probability theory \cite{ash1972probability}. Lemma \ref{lemma:opt} follows from quadratic optimization \cite[Example 4.5]{boyd2004convex}. We omit proofs for brevity. In Lemmas \ref{lem:predictive_variance_helper3}--\ref{lem:predictive_variance_helper2}, $\mathcal{G} = \sigma(h)$ is a sub-$\sigma$-algebra of $\mathcal{F}$ generated by a random vector $h$, and $v : (\Omega, \mathcal{G}) \to (\mathbb{R}^n, \mathcal{B}_{\mathbb{R}^n})$. In Lemma \ref{lem:predictive_variance_helper2}, $w : (\Omega,\mathcal{G}) \rightarrow (\mathbb{R}^m, \mathcal{B}_{\mathbb{R}^m})$.

\begin{lem}[Generalized Peter--Paul inequality]
    \label{lem:peter_paul_inequality}
    Let $y \in \mathbb{R}^n$, $z \in \mathbb{R}^n$, $P \in \mathbb{S}_\geq^n$, and $\alpha > 0$ be given. Then, $2 y^\top P z \leq \alpha^{-1} y^\top P y + \alpha z^\top P z$ holds, and this inequality implies $(y+z)^\top P (y+z) \leq (1 + \alpha^{-1})y^\top P y + (1+\alpha) z^\top P z$.
\end{lem}

\begin{lem}
    \label{lem:predictive_variance_helper3}
    Let $v$ be integrable. Let $M$ be an integrable, $\mathbb{R}^{m \times n}$-valued random matrix, independent of $h$. Then, $M v$ is integrable and $\mathbb{E}(M v | \mathcal{G}) = \mathbb{E}(M) v $ a.e.
\end{lem}

\begin{lem}
    \label{lem:predictive_variance_helper2}
    Let $v$ and $w$ be square-integrable. Let $M$ be an integrable, $\mathbb{R}^{n \times m}$-valued random matrix, independent of $h$. Then, $v^\top M w$ is integrable and $\mathbb{E}(v^\top M w | \mathcal{G}) = v^\top \mathbb{E}(M) w $ a.e.
\end{lem}

\begin{lem}
    \label{lem:predictive_variance_helper1}
    Let $\mathcal{G}$ be a sub-$\sigma$-algebra of $\mathcal{F}$, $M \in \mathbb{R}^{n \times n}$, and $v : (\Omega, \mathcal{F}) \to (\mathbb{R}^n, \mathcal{B}_{\mathbb{R}^n})$ be square-integrable. Then, $\operatorname{tr}(M\operatorname{cov}(v | \mathcal{G})) = \mathbb{E}(v^\top M v |\mathcal{G}) - \mathbb{E}(v | \mathcal{G})^\top M \mathbb{E}(v | \mathcal{G}) $ a.e.
\end{lem}

\begin{lem}
    \label{lemma:opt}
    Define $g : \mathbb{R}^n \times \mathbb{R}^m \rightarrow \mathbb{R}$ as
    \begin{equation*}
        g(x,u) \triangleq \genfrac{[}{]}{0pt}{0}{x}{u}^\top \genfrac{[}{]}{0pt}{0}{\mathmakebox[13.5pt][l]{Q} \; S}{\mathmakebox[13.5pt][l]{S^\top} \; R} \genfrac{[}{]}{0pt}{0}{x}{u} + 2 \genfrac{[}{]}{0pt}{0}{x}{u}^\top \genfrac{[}{]}{0pt}{0}{v}{w} + r
    \end{equation*}
    for some $Q \in \mathbb{S}^n_\geq$, $R \in \mathbb{S}^m_>$, $S \in \mathbb{R}^{n \times m}$, $v \in \mathbb{R}^n$, $w \in \mathbb{R}^m$, and $r \in \mathbb{R}$. Then, for each $x \in \mathbb{R}^n$, there exists a unique $u_x^* \in \mathbb{R}^m$ such that $g(x,u_x^*) = \inf_{u \in \mathbb{R}^m} g(x,u)$. Specifically, $u_x^* = K x + \ell$, where $K \triangleq -R^{-1} S^\top$ and $\ell \triangleq -R^{-1} w$. Furthermore, the minimum value is $g(x,u_x^*) = x^\top(Q - K^\top R K)x + 2 x^\top (v + K^\top w) + r - \ell^\top R \ell$.
\end{lem}

\bibliographystyle{plain}
\bibliography{references.bib}

@article{akella2025risk,
  author={Akella, P. and Dixit, A. and Ahmadi, M. and Lindemann, L. and Chapman, M. P. and Pappas, G. J. and Ames, A. D. and Burdick, J. W.},
  journal={IEEE Control Systems}, 
  title={{Risk-aware robotics: Tail risk measures in planning, control, and verification [Focus on education]}}, 
  year={2025},
  volume={45},
  number={4},
  pages={46--78},
}

@book{ash1972probability,
    author    = {Ash, R. B.},
    title     = {Real Analysis and Probability},
    year      = {1972},
    address   = {New York, NY, USA},
    publisher = {Academic Press},
}

@article{bauerle2011markov,
    author  = {Bäuerle, N. and Ott, J.},
    title   = {Markov decision processes with average-value-at-risk criteria},
    journal = {Mathematical Methods of Operations Research},
    volume  = {74},
    pages   = {361-379},
    year    = {2011},
    doi     = {10.1007/s00186-011-0367-0},
}

@article{bauerle2021minimizing,
    author  = {Bäuerle, N. and Glauner, A.},
    title   = {Minimizing spectral risk measures applied to {M}arkov decision processes},
    journal = {Mathematical Methods of Operations Research},
    volume  = {94},
    number  = {1},
    pages   = {35-69},
    year    = {2021},
    doi     = {10.1007/s00186-021-00746-w},
}

@article{bernstein1989LQG,
    author  = {Bernstein, D. S. and Haddad, W. M.},
    journal = {{IEEE} Transactions on Automatic Control},
    title   = {{LQG} control with an ${H}_\infty$ performance bound: {A} {R}iccati equation approach},
    year    = {1989},
    volume  = {34},
    number  = {3},
    pages   = {293--305},
    doi     = {10.1109/9.16419},
}

@book{bernstein2018scalar,
    author    = {Bernstein, D. S.},
    title     = {Scalar, Vector, and Matrix Mathematics},
    subtitle  = {Theory, Facts, and Formulas - Revised and Expanded},
    year      = {2018},
    publisher = {Princeton University Press},
    address   = {Princeton, NJ, USA},
}

@book{bertsekas1996neuro,
    title     = {Neuro-Dynamic Programming},
    author    = {Bertsekas, D. P. and Tsitsiklis, J. N.},
    year      = {1996},
    publisher = {Athena Scientific},
    address   = {Belmont, MA, USA},
}

@book{bertsekas2019reinforcement,
    author    = {Bertsekas, D. P.},
    title     = {Reinforcement Learning and Optimal Control},
    address   = {Belmont, MA, USA},
    publisher = {Athena Scientific},
    year      = {2019},
}

@book{boyd2004convex,
    author    = {Boyd, S. P. and Vandenberghe, L.},
    title     = {Convex Optimization},
    year      = {2004},
    address   = {Cambridge, UK},
    publisher = {Cambridge University Press},
}

@article{chapman2023on,
    author  = {Chapman, M. P. and Fau{\ss}, M. and Smith, K. M.},
    title   = {On optimizing the Conditional Value-at-Risk of a maximum cost for risk-averse safety analysis},
    journal = {{IEEE} Transactions on Automatic Control},
    volume  = {68},
    number  = {6},
    pages   = {3720--3727},
    year    = {2023},
    doi     = {10.1109/TAC.2022.3195381},
}

@article{chapman2025riskaware,
    author  = {Chapman, M. P. and Kalogerias, D.},
    title   = {Risk-Aware Stability of Linear Systems},
    journal = {{IEEE} Transactions on Automatic Control},
    volume  = {70},
    number  = {2},
    pages   = {861--876},
    year    = {2025},
    doi     = {10.1109/TAC.2024.3444868},
}

@article{fu2005sector,
    author  = {Fu, M. and Xie, L.},
    title   = {The sector bound approach to quantized feedback control},
    journal = {{IEEE} Transactions on Automatic Control},
    volume  = {50},
    number  = {11},
    pages   = {1698-1711},
    year    = {2005},
    doi     = {10.1109/TAC.2005.858689},
}

@article{fujimoto2011optimalfinitetime,
    title   = {Optimal Control of Linear Systems with Stochastic Parameters for Variance Suppression: The Finite Time Horizon Case},
    author  = {Fujimoto, K. and Ogawa, S. and Ota, Y. and Nakayama, M.},
    journal = {IFAC Proceedings Volumes},
    volume  = {44},
    number  = {1},
    pages   = {12605-12610},
    year    = {2011},
    note    = {18th IFAC World Congress},
    doi     = {10.3182/20110828-6-IT-1002.02289},
}

@article{geromel1995convex,
    title   = {A Convex Approach to the Mixed $\mathcal{H}_2$/$\mathcal{H}_\infty$ Control Problem for Discrete-Time Uncertain Systems},
    author  = {Geromel, J. C. and Peres, P. L. D. and Souza, S. R.},
    journal = {{SIAM} Journal on Control and Optimization},
    volume  = {33},
    number  = {6},
    pages   = {1816--1833},
    year    = {1995},
}

@article{habib2023stochastic,
    author  = {Habib, S. and Ahmarinejad, A. and Jia, Y.},
    title   = {A stochastic model for microgrids planning considering smart prosumers, electric vehicles and energy storages},
    journal = {Journal of Energy Storage},
    volume  = {70},
    pages   = {107962},
    year    = {2023},
    doi     = {10.1016/j.est.2023.107962},
}

@book{hernandezlerma2012discrete,
    author    = {Hernández-Lerma, O. and Lasserre, J. B.},
    title     = {Discrete-Time Markov Control Processes},
    year      = {1996},
    publisher = {Springer},
    series    = {Stochastic Modelling and Applied Probability},
}

@article{ito2018risk,
    title   = {Risk-sensitive linear control for systems with stochastic parameters},
    author  = {Ito, Y. and Fujimoto, K. and Tadokoro, Y. and Yoshimura, T.},
    journal = {{IEEE} Transactions on Automatic Control},
    volume  = {64},
    number  = {4},
    pages   = {1328--1343},
    year    = {2019},
    doi     = {10.1109/TAC.2018.2876793},
}

@article{ito2025weighted,
    title   = {Weighted stochastic {R}iccati equations for generalization of linear optimal control},
    author  = {Ito, Y. and Fujimoto, K. and Tadokoro, Y.},
    journal = {Automatica},
    volume  = {171},
    note    = {art. no. 111901},
    year    = {2025},
    doi     = {10.1016/j.automatica.2024.111901},
}

@article{jia2024decentralized,
    title   = {Decentralized stochastic linear-quadratic optimal control with risk constraint and partial observation},
    author  = {Jia, H. and Ni, Y.-H.},
    journal = {Systems \& Control Letters},
    volume  = {187},
    note    = {art. no. 105778},
    year    = {2024},
    doi     = {10.1016/j.sysconle.2024.105778},
}

@book{khalil2002nonlinear,
    author    = {Khalil, H. K.},
    title     = {Nonlinear Systems},
    edition   = {3rd},
    year      = {2002},
    address   = {Upper Saddle River, NJ, USA},
    publisher = {Prentice Hall},
}

@article{kishida2022risk,
    title   = {Risk-aware linear quadratic control using conditional value-at-risk},
    author  = {Kishida, M. and Cetinkaya, A.},
    journal = {{IEEE} Transactions on Automatic Control},
    volume  = {68},
    number  = {1},
    pages   = {416--423},
    year    = {2023},
    doi     = {10.1109/TAC.2022.3142131},
}

@book{kumar1986stochastic,
    author    = {Kumar, P. R. and Varaiya, P.},
    title     = {Stochastic Systems: Estimation, Identification, and Adaptive Control},
    year      = {1986},
    address   = {Englewood Cliffs, NJ, USA},
    publisher = {Prentice Hall},
}

@conference{majumdar2019how,
    author    = {Majumdar, A. and Pavone, M.},
    editor    = {Amato, N. M. and Hager, G. and Thomas, S. and Torres-Torriti, M.},
    title     = {How Should a Robot Assess Risk? {T}owards an Axiomatic Theory of Risk in Robotics},
    booktitle = {Robotics Research},
    series    = {Springer Proceedings in Advanced Robotics},
    volume    = {10},
    number    = {},
    pages     = {75--84},
    year      = {2019},
    doi       = {10.1007/978-3-030-28619-4\_10},
    publisher = {Springer},
    address   = {Cham, CH},
}

@article{patel2024riskaware,
  title = {Risk-Aware Finite-Horizon Social Optimal Control of Mean-Field Coupled Linear-Quadratic Subsystems},
  author = {Patel, D. and Chapman, M. P.},
  year = 2024,
  journal = {{IEEE} Control Systems Letters},
  volume = {8},
  pages = {2265--2270},
  doi = {10.1109/LCSYS.2024.3409456}
}

@mastersthesis{patel2025risk,
    title   = {Risk-Aware Control of Cone-Bounded Nonlinear Systems},
    author  = {Patel, D.},
    year    = {2025},
    school  = {Edward S. Rogers Sr. Department of Electrical and Computer Engineering, University of Toronto},
    address = {Toronto, ON, Canada},
    type    = {{MASc} thesis},
}

@article{royset2025Risk,
    author  = {Royset, J. O.},
    title   = {Risk-Adaptive Approaches to Stochastic Optimization: {A} Survey},
    journal = {SIAM Review},
    volume  = {67},
    number  = {1},
    pages   = {3--70},
    year    = {2025},
    doi     = {10.1137/22M1538946},
}

@article{ruszczynski2010riskaverse,
    author  = {Ruszczyński, A.},
    title   = {Risk-averse dynamic programming for {M}arkov decision processes},
    journal = {Mathematical Programming},
    volume  = {125},
    number  = {2},
    pages   = {235-261},
    year    = {2010},
    doi     = {10.1007/s10107-010-0393-3},
}

@article{scherzinger1982estimation,
    author  = {Scherzinger, B. M. and Kwong, R. H.},
    title   = {Estimation and control of discrete time stochastic systems having cone-bounded non-linearities},
    journal = {International Journal of Control},
    volume  = {36},
    number  = {1},
    pages   = {33-52},
    year    = {1982},
    doi     = {10.1080/00207178208932873},
}

@article{shapiro2012minimax,
    author  = {Shapiro, A.},
    title   = {Minimax and risk averse multistage stochastic programming},
    journal = {European Journal of Operational Research},
    volume  = {219},
    number  = {3},
    pages   = {719-726},
    year    = {2012},
    doi     = {10.1016/j.ejor.2011.11.005},
}

@book{shapiro2021lectures,
    author    = {Shapiro, A. and Dentcheva, D. and Ruszczyński, A.},
    title     = {Lectures on Stochastic Programming: Modeling and Theory},
    edition   = {3rd},
    series    = {{MOS-SIAM} Series on Optimization},
    year      = {2021},
    doi       = {10.1137/1.9781611976595},
    address   = {Philadelphia, PA, USA},
    publisher = {MOS-SIAM},
}

@article{smith2023on,
    author  = {Smith, K. M. and Chapman, M. P.},
    title   = {On Exponential Utility and Conditional Value-at-Risk as Risk-Averse Performance Criteria},
    journal = {{IEEE} Transactions on Control Systems Technology},
    volume  = {31},
    number  = {6},
    pages   = {2555-2570},
    year    = {2023},
    doi     = {10.1109/TCST.2023.3274843},
}

@book{stein2012stochastic,
    author    = {Stein, J. L.},
    title     = {Stochastic Optimal Control and the U.S. Financial Debt Crisis},
    year      = {2012},
    address   = {New York, NY, USA},
    publisher = {Springer},
}

@book{tolstoy2014anna,
    author    = {Tolstoy, L.},
    title     = {Anna Karenina},
    address   = {Oxford, UK},
    publisher = {Oxford University Press},
    note      = {Translated by Rosamund Bartlett},
    year      = {2014},
}

@conference{tsiamis2020riskconstrained,
    author    = {Tsiamis, A. and Kalogerias, D. S. and Chamon, L. F. O. and Ribeiro, A. and Pappas, G. J.},
    title     = {Risk-Constrained Linear-Quadratic Regulators},
    booktitle = {{IEEE} Conference on Decision and Control},
    volume    = {},
    number    = {},
    pages     = {3040-3047},
    year      = {2020},
    doi       = {10.1109/CDC42340.2020.9303967},
}

@inproceedings{vahs2024risk,
    author    = {Vahs, M. and Tumova, J.},
    title     = {Risk-aware Control for Robots with Non-{G}aussian Belief Spaces},
    booktitle = {{IEEE} International Conference on Robotics and Automation},
    year      = {2024},
    volume    = {},
    number    = {},
    pages     = {11661-11667},
    doi       = {10.1109/ICRA57147.2024.10611412},
}

@article{vanParys2016distributionally,
    author  = {Van Parys, B. P. G. and Kuhn, D. and Goulart, P. J. and Morari, M.},
    title   = {Distributionally Robust Control of Constrained Stochastic Systems},
    journal = {{IEEE} Transactions on Automatic Control},
    volume  = {61},
    number  = {2},
    pages   = {430-442},
    year    = {2016},
    doi     = {10.1109/TAC.2015.2444134},
}

@article{wang2022risk,
    author  = {Wang, Y. and Chapman, M. P},
    title   = {Risk-averse autonomous systems: A brief history and recent developments from the perspective of optimal control},
    journal = {Artificial Intelligence},
    volume  = {311},
    number  = {},
    pages   = {103743},
    year    = {2022},
    doi     = {10.1016/j.artint.2022.103743},
}

@article{wang2025policy,
    author  = {Wang, Z. and Gao, Y. and Wang, S. and Zavlanos, M. M. and Abate, A. and Johansson, K. H.},
    journal = {{IEEE} Transactions on Automatic Control},
    title   = {Policy Evaluation in Distributional {LQR}},
    year    = {2025},
    volume  = {70},
    number  = {11},
    pages   = {7477-7492},
    doi     = {10.1109/TAC.2025.3575649},
}

@article{weitobesubmitted,
    author  = {Wei, C. and Li, K. F. and Kalogerias, D. and Chapman, M. P.},
    title   = {Risk-aware linear-quadratic regulation with
               temporally coupled states},
    journal = {\emph{Submitted to} Automatica\emph{, preprint, arXiv:2603.23737}},
    year    = {unpublished results},
}

@article{whittle1981risk,
    author  = {Whittle, P.},
    title   = {{Risk-Sensitive Linear/Quadratic/Gaussian Control}},
    journal = {Advances in Applied Probability},
    volume  = {13},
    number  = {4},
    pages   = {764--777},
    year    = {1981},
    doi     = {10.2307/1426972},
}

@article{zhang2021policy,
    title   = {{Policy optimization for $\mathcal{H}_2$ linear control with $\mathcal{H}_\infty$ robustness guarantee: Implicit regularization and global convergence}},
    author  = {Zhang, K. and Hu, B. and Ba\c{s}ar, T.},
    journal = {{SIAM} Journal on Control and Optimization},
    volume  = {59},
    number  = {6},
    pages   = {4081--4109},
    year    = {2021},
    doi     = {10.1137/20M1347942},
}

@book{zhou1996robust,
    author    = {Zhou, K. and Doyle, J. C. and Glover, K.},
    title     = {Robust and Optimal Control},
    year      = {1996},
    publisher = {Prentice Hall},
    addresss  = {Upper Saddle River, NJ, USA},
}

\end{document}